%% file: NegativeSimCoop.tex
\documentclass{llncs}
\usepackage[font={small}]{caption}
\usepackage{makeidx}  
\usepackage{amsmath}
\usepackage{amssymb}
\usepackage{amsfonts}
\usepackage{graphicx}
\usepackage{color}
\usepackage{subfig}
\usepackage{hyperref}
\usepackage{cite}
\usepackage{comment}
\usepackage{wrapfig}
\usepackage{transparent}
\usepackage{appendix}
\usepackage[vlined]{algorithm2e}

\pdfpageattr {/Group << /S /Transparency /I true /CS /DeviceRGB>>}
\newcommand{\fullgridgraph}{G^\mathrm{f}}\newcommand{\bindinggraph}{G^\mathrm{b}}

\newif\ifabstract
\newif\iffull

\abstracttrue

\ifabstract
	\fullfalse
\else
	\fulltrue
\fi

\newtoks\magicAppendix
\magicAppendix={}
\newtoks\magictoks
\newif\iflater
\laterfalse
\ifabstract
\long\def\later#1{\magictoks={#1}%
  \edef\magictodo{\noexpand\magicAppendix={\the\magicAppendix \par
    \the\magictoks%
  }}
  \magictodo}
\long\def\both#1{\magictoks={#1}%
  \edef\magictodo{\noexpand\magicAppendix={\the\magicAppendix \par
    \noexpand\setcounter{theorem-preserve}{\noexpand\arabic{theorem}}%
    \noexpand\setcounter{theorem}{\arabic{theorem}}%
    \noexpand\setcounter{section-preserve}{\noexpand\arabic{section}}%
    \noexpand\setcounter{section}{\arabic{section}}%
	\noexpand\let\noexpand\oldsection=\noexpand\thesection
	\noexpand\def\noexpand\thesection{\thesection}
	\noexpand\let\noexpand\oldlabel=\noexpand\label
	\noexpand\let\noexpand\label=\noexpand\blank
    \the\magictoks%
    \noexpand\setcounter{theorem}{\noexpand\arabic{theorem-preserve}}%
    \noexpand\setcounter{section}{\noexpand\arabic{section-preserve}}%
	\noexpand\let\noexpand\thesection=\noexpand\oldsection
	\noexpand\let\noexpand\label=\noexpand\oldlabel
  }}
  \magictodo
  \the\magictoks}
\else
\long\def\later#1{#1}
\long\def\both#1{#1}
\fi
\long\def\magicappendix{
	\latertrue%
	\the\magicAppendix%
}

\title{Doubles and Negatives are Positive (in Self-Assembly)}
\vspace{-10pt}
\institute{}

\author{
  Jacob Hendricks%
        \thanks{Department of Computer Science and Computer Engineering, University of Arkansas,
      \protect\url{jhendric@uark.edu}
      Supported in part by National Science Foundation Grant CCF-1117672.}
\and
  Matthew J. Patitz%
    \thanks{Department of Computer Science and Computer Engineering, University of Arkansas,
      \protect\url{patitz@uark.edu}
      Supported in part by National Science Foundation Grant CCF-1117672.}
\and
 Trent A. Rogers%
        \thanks{Department of Mathematical Sciences, University of Arkansas,
      \protect\url{tar003@uark.edu}
      Supported in part by National Science Foundation Grant CCF-1117672.}
      }
\date{}

\input{commands-tam.sty}

\begin{document}

\maketitle
\vspace{-18pt}
\begin{abstract}
In the abstract Tile Assembly Model (aTAM), the phenomenon of cooperation occurs when the attachment of a new tile to a growing assembly requires it to bind to more than one tile already in the assembly.  Often referred to as ``temperature-2'' systems, those which employ cooperation are known to be quite powerful (i.e. they are computationally universal and can build an enormous variety of shapes and structures).  Conversely, aTAM systems which do not enforce cooperative behavior, a.k.a. ``temperature-1'' systems, are conjectured to be relatively very weak, likely to be unable to perform complex computations or algorithmically direct the process of self-assembly.  Nonetheless, a variety of models based on slight modifications to the aTAM have been developed in which temperature-1 systems are in fact capable of Turing universal computation through a restricted notion of cooperation.  Despite that power, though, several of those models have previously been proven to be unable to perform or simulate the stronger form of cooperation exhibited by temperature-2 aTAM systems.

In this paper, we first prove that another model in which temperature-1 systems are computationally universal, namely the restricted glue TAM (rgTAM) in which tiles are allowed to have edges which exhibit repulsive forces, is also unable to simulate the strongly cooperative behavior of the temperature-2 aTAM.  We then show that by combining the properties of two such models, the Dupled Tile Assembly Model (DTAM) and the rgTAM into the DrgTAM, we derive a model which is actually more powerful at temperature-1 than the aTAM at temperature-2.  Specifically, the DrgTAM, at temperature-1, can simulate any aTAM system of any temperature, and it also contains systems which cannot be simulated by any system in the aTAM.
\end{abstract}
\vspace{-15pt}

\input{intro}

\input{prelims}

\input{rgTAS_cannot_coop}
\input{simCoop}
\vspace{-15pt}
\bibliographystyle{plain}
\bibliography{tam}

\ifabstract
\newpage
\appendix
\magicappendix
\fi

\end{document}

%% file: intro.tex
\vspace{-20pt}
\section{Introduction}
\label{sec:intro}
\vspace{-10pt}

Composed of large collections of relatively simple components which autonomously combine to form predetermined structures, self-assembling systems provide a framework in which structures can grow from the bottom up, with precise placement of individual molecules.  Natural self-assembling systems, the results of which include structures ranging from crystalline snowflakes to cellular membranes and viruses, have inspired a large body of research focused on both studying their properties and creating artificial self-assembling systems to mimic them.  As experimental and theoretical research into self-assembly has increased in sophistication, particular attention has been focused upon the domain of \emph{algorithmic self-assembly}, which is self-assembly intrinsically directed by algorithms, or step-by-step procedures used to perform computations.  An example of a model supporting algorithmic self-assembly is the abstract Tile Assembly Model (aTAM) \cite{Winf98}, which has spawned much research investigating its powers and limitations, and even more fundamentally those of algorithmic self-assembly in general.

In the aTAM, the fundamental components are square \emph{tiles} which have sticky \emph{glues} on the edges which allow them to bind with other tiles along edges sharing matching glues.  Self-assembly begins from special \emph{seed} assemblies, and progresses as tiles attach one at a time to the growing assembly.  As simple as the aTAM sounds, when initially introducing it in 1998 \cite{Winf98}, Winfree showed it be to capable of Turing universal computation, i.e. it can perform any computation possible by any computer.  It was soon also shown that the algorithmic nature of the aTAM can be harnessed to build squares \cite{RotWin00} and general shapes \cite{SolWin07} with (information theoretically) optimal efficiency in terms of the number of unique kinds of tiles used in the assemblies.  The rich set of results displaying the power of the aTAM (e.g. \cite{IUSA,jCCSA,jSADS} to name just a few), however, have appeared to be contingent upon a minimal value of $2$ for a system parameter known as the \emph{temperature}.  The temperature of an aTAM system is the threshold which, informally stated, determines how many glues a tile must bind to a growing assembly with in order to remain attached.  Temperature-$2$ systems have the property that they can enforce \emph{cooperation} in which the attachment of a tile requires it to correctly bind to at least two tiles already in the assembly (thus, those two tiles \emph{cooperate} to allow the new tile to attach). This cooperation allows for each tile to effectively perform a primitive logical operation (e.g. \texttt{and}, \texttt{or}, \texttt{xor}, etc.) on the ``input'' values supplied by the tiles they bind to, and careful combination of these operations, just as with the gates in a modern electronic processor, allow for complex computations to occur.  In contrast, the requirement for cooperation cannot be enforced in temperature-$1$ systems which only require one binding side, and it has thus been conjectured that temperature-$1$ aTAM systems are ``weak'' in the sense that they cannot perform universal computation or be guided algorithmically \cite{jLSAT1}.  While this long-standing conjecture remains unproven in the general case of the aTAM, a growing body of work has focused on attempts to circumvent the limitations of temperature-$1$ self-assembly by making small variations to the aTAM.  For instance, it has been shown that the following models are computationally universal at temperature-$1$: the 3-D aTAM \cite{CooFuSch11}, aTAM systems which compute probabilistically \cite{CooFuSch11}, the restricted glues TAM (rgTAM) which allow glues with repulsive (rather than just attractive) forces \cite{SingleNegative}, the Dupled aTAM which allows tiles shaped like $2 \times 1$ rectangles \cite{Duples}, and the Signal-passing Tile Assembly Model \cite{Signals} which contains dynamically reconfigurable tiles.

While such results may seem to indicate that those computationally universal models are as powerful as the temperature-$2$ aTAM, in \cite{IUNeedsCoop} it was shown that 3-D temperature-$1$ aTAM systems cannot possibly simulate very basic ``glue cooperation'' exhibited in the temperature-$2$ aTAM where a new tile actually binds to two already placed tiles.  Essentially, the weaker form of cooperation exploited by the 3-D temperature-$1$ aTAM to perform computation does allow for the restriction of tile placements based on the prior placement of two other tiles, but that form of cooperation seems to be fundamentally restrictive and ``non-additive'', meaning that the previously placed tiles can only prevent certain future tile bindings, but not cooperate to support new binding possibilities.  In fact, that lesser form of cooperation now appears to be the limit for those temperature-$1$ models which can compute (with perhaps the exception of the active signal-passing tiles), as it was shown in \cite{Duples} that the DaTAM also cannot simulate glue cooperation.  It appears that the landscape modeling the relative powers of models across various parameters is more subtle and complicated than originally recognized, with the original notion of cooperative behavior being more refined.

The contributions of this paper are threefold.  First, we show that the rgTAM is also not capable of simulating glue cooperation.  Second, we introduce the Dupled restricted glue TAM (DrgTAM) which allows for both square tiles and ``duple'' tiles, which are simply pre-formed pairs of $2$ tiles joined along one edge before assembly begins, and it allows for glues with negative strength (i.e. those which exert repulsive force).  However, it is restricted similar to the rgTAM in that the magnitude of glue strengths cannot exceed $1$ (i.e. only strengths $1$ and $-1$ are allowed).  Third, we show that by creating the DrgTAM by combining two models (the rgTAM and the Dupled aTAM) which are computationally universal at temperature $1$ but which cannot independently simulate glue cooperation, the result is a model which in some measures is greater than the sum of its parts.  That is, the resulting DrgTAM is capable of both universal computation \emph{and} the simulation of glue cooperation.  This is the first such result for passive (i.e. non-active) tile assembly systems.  In fact, we show the stronger result that there is a single tile set in the DrgTAM which can be configured to, in a temperature-$1$ system, simulate any arbitrary aTAM system, making it intrinsically universal for the aTAM.  Coupled with the result in \cite{Duples} which proves that there are temperature-$1$ systems in the DTAM, which are thus also in the DrgTAM, that cannot be simulated by the aTAM at temperature-$2$, this actually implies that the DrgTAM is more powerful than the temperature-$2$ aTAM.  

The paper is organized as follows.  In Section~\ref{sec:prelims} we give high-level sketches of the definitions of the models and of the concepts of simulation used throughout the paper.  In Section~\ref{rgTAS_cannot_coop} we prove that rgTAM systems cannot simulate the glue cooperation of temperature-$2$ aTAM systems, and in Section~\ref{sec:DrgTAM_sim} we present the proof that the DrgTAM can simulate the temperature-$2$ aTAM and in fact contains a tile set which is intrinsically universal for it.  Due to space constraints, the formal definitions as well as all proofs can be found in the Appendix.

%% file: prelims.tex
\vspace{-15pt}
\section{Preliminaries}\label{sec:prelims}

\input{tam-informal}
\input{tam-formal}
\input{simulation_def}

%% file: tam-informal.tex
\vspace{-10pt}
 Throughout this paper, we use three tile assembly models: 1. the aTAM, 2. the restricted glue TAM (rgTAM), and 3. the dupled rgTAM (DrgTAM). We now informally describe these models.  For formal definitions see Section~\ref{sec-tam-formal}.
\vspace{-10pt}
\subsubsection{Informal description of the abstract Tile Assembly Model}
\label{sec-tam-informal}

A \emph{tile type} is a unit square with four sides, each consisting of a \emph{glue label}, often represented as a finite string, and a nonnegative integer \emph{strength}. A glue~$g$ that appears on multiple tiles (or sides) always has the same strength~$s_g$.
There are a finite set $T$ of tile types, but an infinite number of copies of each tile type, with each copy being referred to as a \emph{tile}. An \emph{assembly}
is a positioning of tiles on the integer lattice $\Z^2$, described  formally as a partial function $\alpha:\Z^2 \dashrightarrow T$.
Let $\mathcal{A}^T$ denote the set of all assemblies of tiles from $T$, and let $\mathcal{A}^T_{< \infty}$ denote the set of finite assemblies of tiles from $T$.
We write $\alpha \sqsubseteq \beta$ to denote that $\alpha$ is a \emph{subassembly} of $\beta$, which means that $\dom\alpha \subseteq \dom\beta$ and $\alpha(p)=\beta(p)$ for all points $p\in\dom\alpha$.
Two adjacent tiles in an assembly \emph{interact}, or are \emph{attached}, if the glue labels on their abutting sides are equal and have positive strength.
Each assembly induces a \emph{binding graph}, a grid graph whose vertices are tiles, with an edge between two tiles if they interact.
The assembly is \emph{$\tau$-stable} if every cut of its binding graph has strength at least~$\tau$, where the strength   of a cut is the sum of all of the individual glue strengths in the cut. When $\tau$ is clear from context, we simply say that a $\tau$-stable assembly is stable.

A \emph{tile assembly system} (TAS) is a triple $\calT = (T,\sigma,\tau)$, where $T$ is a finite set of tile types, $\sigma:\Z^2 \dashrightarrow T$ is a finite, $\tau$-stable \emph{seed assembly},
and $\tau$ is the \emph{temperature}.
An assembly $\alpha$ is \emph{producible} if either $\alpha = \sigma$ or if $\beta$ is a producible assembly and $\alpha$ can be obtained from $\beta$ by the stable binding of a single tile.
In this case we write $\beta\to_1^\calT \alpha$ (to mean~$\alpha$ is producible from $\beta$ by the attachment of one tile), and we write $\beta\to^\calT \alpha$ if $\beta \to_1^{\calT*} \alpha$ (to mean $\alpha$ is producible from $\beta$ by the attachment of zero or more tiles).
When $\calT$ is clear from context, we may write $\to_1$ and $\to$ instead.
We let $\prodasm{\calT}$ denote the set of producible assemblies of $\calT$.
An assembly is \emph{terminal} if no tile can be $\tau$-stably attached to it.
We let   $\termasm{\calT} \subseteq \prodasm{\calT}$ denote  the set of producible, terminal assemblies of $\calT$.
A TAS $\calT$ is \emph{directed} if $|\termasm{\calT}| = 1$. Hence, although a directed system may be nondeterministic in terms of the order of tile placements,  it is deterministic in the sense that exactly one terminal assembly is producible (this is analogous to the notion of {\em confluence} in rewriting systems).

Since the behavior of a TAS $\calT=(T,\sigma,\tau)$ is unchanged if every glue with strength greater than $\tau$ is changed to have strength exactly $\tau$, we assume  that all glue strengths are in the set $\{0, 1, \ldots , \tau\}$.
\vspace{-15pt}
\subsubsection{Informal description of the restricted glue Tile Assembly Model}
\label{sec-rgtam-informal}

The rgTAM was introduced in~\cite{SingleNegative} where it was shown that the rgTAM is computationally universal even in the case where only a single glue has strength $-1$. The definition used in~\cite{SingleNegative} and the definition given here are similar to the irreversible negative glue tile assembly model given in~\cite{DotKarMasNegativeJournal}.

The restricted glue Tile Assembly Model (rgTAM) can be thought of as the aTAM where the temperature is restricted to $1$ and glues may have strengths $-1, 0,$ or $1$. A system in the rgTAM is an ordered pair $(T, \sigma)$ where $T$ is the \emph{tile set}, and $\sigma$ is a stable \emph{seed assembly}. We call an rgTAM system an rgTAS. \emph{Producible} assemblies in an rgTAS can be defined recursively as follows. Let $\mathcal{T} = (T,\sigma)$ be an rgTAS. Then, an assembly $\alpha$ is producible in $\mathcal{T}$ if 1. $\alpha = \sigma$, 2. $\alpha$ is the result of a stable attachment of a single tile to a producible assembly, or 3. $\alpha$ is one side of a cut of strength $\leq 0$ of a producible assembly.

In~\cite{DotKarMasNegativeJournal}, Doty et al. give a list of the choices that can be made when defining a model with negative glues. These choices are (1) seeded/unseeded, (2) single-tile addition/two-handed assembly, (3) irreversible/reversible, (4) detachment precedes attachment/detachment and attachment in arbitrary order, (5) finite tile counts/infinite tile counts, and (6) tagged result/tagged junk. Here we have chosen the rgTAM to be a seeded, single-tile addition, irreversible model that uses infinite tiles. We also assume that attachment and detachment in the model occur in arbitrary order, however the results presented here also hold in the case where detachment precedes attachment. Finally, the definition of simulation (see Section~\ref{sec:simulation_def_informal}) implicitly enforces a notion of tagged result and tagged junk. In particular, if detachment occurs in a simulating system, of the two resulting assemblies one contains the seed and represents some assembly in the simulated system, while the other resulting assembly must map to the empty tile.

\vspace{-15pt}
\subsubsection{Informal description of the Dupled restricted glue Tile Assembly Model}
\label{sec-drgtam-informal}

The DrgTAM is an extension of the rgTAM which allows for systems with square tiles as well as rectangular tiles. The rectangular tiles are $2 \times 1$ or $1 \times 2$ rectangles which can logically be thought of as two square tiles which begin pre-attached to each other along an edge, hence the name \emph{duples}.  A \emph{DrgTAM system} (DrgTAS) is an ordered 4-tuple $(T,S,D,\sigma)$ where, as in a TAS, $T$ is a tile set and $\sigma$ is a seed assembly. $S$ is the set of singleton (i.e. square) tiles which are available for assembly, and $D$ is the set of duple tiles.  The tile types making up $S$ and $D$ all belong to $T$, with those in $D$ each being a combination of two tile types from $T$.

It should be noted that the glue binding two tiles that form a duple must have strength $1$, and the glues exposed by a duple may have strength $-1$, $0$, or $1$. Also notice that for an assembly $\alpha$ in a DrgTAS, a cut of strength $\leq 0$ may separate two nodes of the grid graph that correspond to two tiles of a duple. Then, the two producible assemblies on each side of this cut each contain one tile from the duple.

%% file: tam-formal.tex
\ifabstract
\later{
\section{Formal descriptions of the Tile Assembly Models}

We now give the formal definitions of the tile assembly models.

\subsection{Formal description of the abstract Tile Assembly Model}
\label{sec-tam-formal}

In this section we provide a set of definitions and conventions that are used throughout this paper.

We work in the $2$-dimensional discrete space $\Z^2$. Define the set
$U_2=\{(0,1), \\(1,0), (0,-1), (-1,0)\}$ to be the set of all
\emph{unit vectors} in $\mathbb{Z}^2$.
We also sometimes refer to these vectors by their
cardinal directions $N$, $E$, $S$, $W$, respectively.
All \emph{graphs} in this paper are undirected.
A \emph{grid graph} is a graph $G =
(V,E)$ in which $V \subseteq \Z^2$ and every edge
$\{\vec{a},\vec{b}\} \in E$ has the property that $\vec{a} - \vec{b} \in U_2$.

Intuitively, a tile type $t$ is a unit square that can be
translated, but not rotated, having a well-defined ``side
$\vec{u}$'' for each $\vec{u} \in U_2$. Each side $\vec{u}$ of $t$
has a ``glue'' with ``label'' $\textmd{label}_t(\vec{u})$--a string
over some fixed alphabet--and ``strength''
$\textmd{str}_t(\vec{u})$--a nonnegative integer--specified by its type
$t$. Two tiles $t$ and $t'$ that are placed at the points $\vec{a}$
and $\vec{a}+\vec{u}$ respectively, \emph{bind} with \emph{strength}
$\textmd{str}_t\left(\vec{u}\right)$ if and only if
$\left(\textmd{label}_t\left(\vec{u}\right),\textmd{str}_t\left(\vec{u}\right)\right)
=
\left(\textmd{label}_{t'}\left(-\vec{u}\right),\textmd{str}_{t'}\left(-\vec{u}\right)\right)$.

In the subsequent definitions, given two partial functions $f,g$, we write $f(x) = g(x)$ if~$f$ and~$g$ are both defined and equal on~$x$, or if~$f$ and~$g$ are both undefined on $x$.

Fix a finite set $T$ of tile types.
A $T$-\emph{assembly}, sometimes denoted simply as an \emph{assembly} when $T$ is clear from the context, is a partial
function $\pfunc{\alpha}{\Z^2}{T}$ defined on at least one input, with points $\vec{x}\in\Z^2$ at
which $\alpha(\vec{x})$ is undefined interpreted to be empty space,
so that $\dom \alpha$ is the set of points with tiles.

We write $|\alpha|$ to denote $|\dom \alpha|$, and we say $\alpha$ is
\emph{finite} if $|\alpha|$ is finite. For assemblies $\alpha$
and $\alpha'$, we say that $\alpha$ is a \emph{subassembly} of
$\alpha'$, and write $\alpha \sqsubseteq \alpha'$, if $\dom \alpha
\subseteq \dom \alpha'$ and $\alpha(\vec{x}) = \alpha'(\vec{x})$ for
all $x \in \dom \alpha$.

We now give a brief formal definition of the aTAM. See \cite{Winf98,RotWin00,Roth01,jSSADST} for other developments of the model.  Our notation is that of \cite{jSSADST}, which also contains a more complete definition.

Given a set $T$ of tile types, an {\it assembly} is a partial function $\pfunc{\alpha}{\mathbb{Z}^2}{T}$. An assembly is {\it $\tau$-stable}
if it cannot be broken up into smaller assemblies without breaking bonds of total strength at least $\tau$, for some $\tau \in \mathbb{N}$.

Self-assembly begins with a {\it seed assembly} $\sigma$ and
proceeds asynchronously and nondeterministically, with tiles
adsorbing one at a time to the existing assembly in any manner that
preserves $\tau$-stability at all times.  A {\it tile assembly system}
({\it TAS}) is an ordered triple $\mathcal{T} = (T, \sigma, \tau)$,
where $T$ is a finite set of tile types, $\sigma$ is a seed assembly
with finite domain, and $\tau \in \N$.  A {\it generalized tile
assembly system} ({\it GTAS})
is defined similarly, but without the finiteness requirements.  We
write $\prodasm{\mathcal{T}}$ for the set of all assemblies that can arise
(in finitely many steps or in the limit) from $\mathcal{T}$.  An
assembly $\alpha \in \prodasm{\mathcal{T}}$ is {\it terminal}, and we write $\alpha \in
\termasm{\mathcal{T}}$, if no tile can be $\tau$-stably added to it. It is clear that $\termasm{\mathcal{T}} \subseteq \prodasm{\mathcal{T}}$.

An assembly sequence in a TAS $\mathcal{T}$ is a (finite or infinite) sequence $\vec{\alpha} = (\alpha_0,\alpha_1,\ldots)$ of assemblies in which each $\alpha_{i+1}$ is obtained from $\alpha_i$ by the addition of a single tile. The \emph{result} $\res{\vec{\alpha}}$ of such an assembly sequence is its unique limiting assembly. (This is the last assembly in the sequence if the sequence is finite.) The set $\prodasm{T}$ is partially ordered by the relation $\longrightarrow$ defined by
\begin{eqnarray*}
\alpha \longrightarrow \alpha' & \textmd{iff} & \textmd{there is an assembly sequence } \vec{\alpha} = (\alpha_0,\alpha_1,\ldots) \\
                               &              & \textmd{such that } \alpha_0 = \alpha \textmd{ and } \alpha' = \res{\vec{\alpha}}. \\
\end{eqnarray*}
If $\vec{\alpha} = (\alpha_0,\alpha_1,\ldots)$ is an assembly sequence in $\mathcal{T}$ and $\vec{m} \in \mathbb{Z}^2$, then the $\vec{\alpha}$\emph{-index} of $\vec{m}$ is $i_{\vec{\alpha}}(\vec{m}) = $min$\{ i \in \mathbb{N} | \vec{m} \in \dom \alpha_i\}$.  That is, the $\vec{\alpha}$-index of $\vec{m}$ is the time at which any tile is first placed at location $\vec{m}$ by $\vec{\alpha}$.  For each location $\vec{m} \in \bigcup_{0 \leq i \leq l} \dom \alpha_i$, define the set of its input sides IN$^{\vec{\alpha}}(\vec{m}) = \{\vec{u} \in U_2 | \mbox{str}_{\alpha_{i_{\alpha}}(\vec{m})}(\vec{u}) > 0 \}$.

We say that $\mathcal{T}$ is \emph{directed} (a.k.a. \emph{deterministic}, \emph{confluent}, \emph{produces a unique assembly}) if the relation $\longrightarrow$ is directed, i.e., if for all $\alpha,\alpha' \in \prodasm{T}$, there exists $\alpha'' \in \prodasm{T}$ such that $\alpha \longrightarrow \alpha''$ and $\alpha' \longrightarrow \alpha''$. It is easy to show that $\mathcal{T}$ is directed if and only if there is a unique terminal assembly $\alpha \in \prodasm{T}$ such that $\sigma \longrightarrow \alpha$.

A set $X \subseteq \Z^2$ {\it weakly self-assembles} if there exists
a TAS ${\mathcal T} = (T, \sigma, \tau)$ and a set $B \subseteq T$
such that $\alpha^{-1}(B) = X$ holds for every terminal assembly
$\alpha \in \termasm{T}$.  Essentially, weak self-assembly can be thought of
as the creation (or ``painting'') of a pattern of tiles from $B$ (usually taken to be a
unique ``color'') on a possibly larger ``canvas'' of un-colored tiles.

A set $X$ \emph{strictly self-assembles} if there is a TAS $\mathcal{T}$ for
which every assembly $\alpha\in\termasm{T}$ satisfies $\dom \alpha =
X$. Essentially, strict self-assembly means that tiles are only placed
in positions defined by the shape.  Note that if $X$ strictly self-assembles, then $X$ weakly
self-assembles. (Let all tiles be in $B$.)

\subsection{Formal description of the restricted glue Tile Assembly Model}
\label{sec-rgtam-formal}

In this section we formally define the restricted glue Tile Assembly Model (rgTAM). Since the rgTAM is based on the aTAM, most of the formal definition of Section~\ref{sec-tam-formal} apply here.
The rgTAM can be thought of as the aTAM where every system (rgTAS) in the rgTAM has the properties that $\tau = 1$ and glues may have strengths $-1, 0,$ or $1$. A system in the rgTAM is defined as an ordered pair $(T, \sigma)$ where $T$ is a set of tile types, and $\sigma$ is a stable seed assembly.

An assembly sequence in an rgTAS $\mathcal{T}$ is a (finite or infinite) sequence $\vec{\alpha} = (\alpha_0,\alpha_1,\ldots)$ of assemblies in which each $\alpha_{i+1}$ is obtained from $\alpha_i$
in one of two ways. First, $\alpha_{i+1}$ can obtained from $\alpha_i$ by the addition of a single tile such that the sum of the strengths of bound glues of this single tile in $\alpha_{i+1}$ is $\geq 1$. Second, $\alpha_{i+1}$ can obtained from $\alpha_i$ if $\alpha_{i+1}$ lies on one side of a cut of $\alpha_{i}$ such that the strength of this cut is $\leq 0$. Unlike an assembly sequence in the aTAM, assembly sequences in the rgTAM may not have a unique limiting assembly, and therefore, may not have a result. However, given an assembly $\alpha$ in an rgTAS, and an assembly sequence $\vec{\alpha}$ if the limit of $\vec{\alpha}$ is $\alpha$, then we say that the \emph{result} (denoted $\res{\vec{\alpha}}$) of $\vec{\alpha}$ is $\alpha$. The notations used in Section~\ref{sec-tam-formal} apply to the rgTAM. In addition to these notations, we distinguish between tile attachment and assemblies produced by a cut of strength $\leq 0$ as follows.

\begin{eqnarray*}
\alpha \rightarrow_+ \alpha' & \textmd{iff} & \alpha' \textmd{ is obtained from } \alpha \textmd{ by a single stable tile addition } \\
\alpha \rightarrow_- (\alpha'_1, \alpha'_2) & \textmd{iff} & \alpha'_1 \textmd{ and } \alpha'_2 \textmd{ lie on each side of a cut of } \alpha \textmd{ such that the } \\
&&\textmd{ strength of this cut is } \leq 0
\end{eqnarray*}

\subsection{Formal description of the Dupled restricted glues Tile Assembly Model}
\label{sec-Drgtam-formal}

This section gives a formal definition of the Dupled restricted glues Tile Assembly Model (DrgTAM).
First, we define the dupled aTAM (DaTAM), which is a mild extension of Winfree's abstract tile assembly model \cite{Winf98}. Then, as in~\ref{sec-rgtam-formal}, we define the DrgTAM by restricting temperature to $1$ and glues strengths to $-1$, $0$, or $1$.

Given $V \subseteq \Z^2$, the \emph{full grid graph} of $V$ is the undirected graph $\fullgridgraph_V=(V,E)$,
and for all $\vec{x}, \vec{y}\in V$, $\left\{\vec{x},\vec{y}\right\} \in E \iff \| \vec{x} - \vec{y}\| = 1$; i.e., if and only if $\vec{x}$ and $\vec{y}$ are adjacent on the $2$-dimensional integer Cartesian space. Fix an alphabet $\Sigma$. $\Sigma^*$ is the set of finite strings over $\Sigma$. Let $\Z$, $\Z^+$, and $\N$ denote the set of integers, positive integers, and nonnegative integers, respectively.

A \emph{square tile type} is a tuple $t \in (\Sigma^* \times \N)^{4}$; i.e. a unit square, with four sides, listed in some standardized order, and each side having a \emph{glue} $g \in \Sigma^* \times \N$ consisting of a finite string \emph{label} and nonnegative integer \emph{strength}. Let $T\subseteq (\Sigma^* \times \N)^{4}$ be a set of tile types. We define a set of \emph{singleton types} to be any subset $S \subseteq T$. Let $t = ((g_N,s_N),(g_S,s_S),(g_E,s_E),(g_W,s_W)) \in T$, $d\in \{N,S,E,W\} = \mathcal{D}$, and write $Glue_d(t) = g_d$ and $Strength_d(t) = s_d$. A \emph{duple type} is defined as an element of the set
$\{ (x,y,d) \mid x,y\in T, \; d\in\mathcal{D}, \; Glue_d(x) = Glue_{-d}(y), \; \textmd{ and }Strength_d(x)=Strength_{-d}(y)\geq\tau \}$.

A {\em configuration} is a (possibly empty) arrangement of tiles on the integer lattice $\Z^2$, i.e., a partial function $\alpha:\Z^2 \dashrightarrow T$.
Two adjacent tiles in a configuration \emph{interact}, or are \emph{attached}, if the glues on their abutting sides are equal (in both label and strength) and have positive strength.
Each configuration $\alpha$ induces a \emph{binding graph} $\bindinggraph_\alpha$, a grid graph whose vertices are positions occupied by tiles, according to $\alpha$, with an edge between two vertices if the tiles at those vertices interact. An \emph{assembly} is a connected, non-empty configuration, i.e., a partial function $\alpha:\Z^2 \dashrightarrow T$ such that $\fullgridgraph_{\dom \alpha}$ is connected and $\dom \alpha \neq \emptyset$. The \emph{shape} $S_\alpha \subseteq \Z^d$ of $\alpha$ is $\dom \alpha$. Let $\alpha$ be an assembly and $B \subseteq
\mathbb{Z}^2$. $\alpha$ \emph{restricted to} $B$, written as $\alpha
\upharpoonright B$, is the unique assembly satisfying $\left(\alpha
\upharpoonright B\right) \sqsubseteq \alpha$, and $\dom{\left(\alpha
\upharpoonright B\right)} = B$

Given $\tau\in\Z^+$, $\alpha$ is \emph{$\tau$-stable} if every cut of~$\bindinggraph_\alpha$ has weight at least $\tau$, where the weight of an edge is the strength of the glue it represents. When $\tau$ is clear from context, we say $\alpha$ is \emph{stable}.
Given two assemblies $\alpha,\beta$, we say $\alpha$ is a \emph{subassembly} of $\beta$, and we write $\alpha \sqsubseteq \beta$, if $S_\alpha \subseteq S_\beta$ and, for all points $p \in S_\alpha$, $\alpha(p) = \beta(p)$. Let $\mathcal{A}^T$ denote the set of all assemblies of tiles from $T$, and let $\mathcal{A}^T_{< \infty}$ denote the set of finite assemblies of tiles from $T$.

A \emph{dupled tile assembly system} (DTAS) is a tuple $\mathcal{T} = (T,S,D,\sigma,\tau)$, where $T$ is a finite tile set, $S \subseteq T$ is a finite set of singleton types, $D$ is a finite set of duple tile types, $\sigma:\Z^2 \dashrightarrow T$ is the finite, $\tau$-stable, \emph{seed assembly}, and $\tau\in\Z^+$ is the \emph{temperature}.

Given two $\tau$-stable assemblies $\alpha,\beta$, we write $\alpha \to_1^{\mathcal{T}} \beta$ if $\alpha \sqsubseteq \beta$, $0 < |S_{\beta} \setminus S_{\alpha}| \leq 2$. In this case we say $\alpha$ \emph{$\mathcal{T}$-produces $\beta$ in one step}. The \emph{$\mathcal{T}$-frontier} of $\alpha$ is the set $\partial^\mathcal{T} \alpha = \bigcup_{\alpha \to_1^\mathcal{T} \beta} S_{\beta} \setminus S_{\alpha}$, the set of empty locations at which a tile could stably attach to $\alpha$.

A sequence of $k\in\Z^+ \cup \{\infty\}$ assemblies $\alpha_0,\alpha_1,\ldots$ over $\mathcal{A}^T$ is a \emph{$\mathcal{T}$-assembly sequence} if, for all $1 \leq i < k$, $\alpha_{i-1} \to_1^\mathcal{T} \alpha_{i}$.
The {\em result} of an assembly sequence is the unique limiting assembly (for a finite sequence, this is the final assembly in the sequence).
If $\vec{\alpha} = (\alpha_0,\alpha_1,\ldots)$ is an assembly sequence in $\mathcal{T}$ and $\vec{m} \in \mathbb{Z}^2$, then the $\vec{\alpha}$\emph{-index} of $\vec{m}$ is $i_{\vec{\alpha}}(\vec{m}) = $min$\{ i \in \mathbb{N} | \vec{m} \in \dom \alpha_i\}$.  That is, the $\vec{\alpha}$-index of $\vec{m}$ is the time at which any tile is first placed at location $\vec{m}$ by $\vec{\alpha}$.  For each location $\vec{m} \in \bigcup_{0 \leq i \leq l} \dom \alpha_i$, define the set of its input sides IN$^{\vec{\alpha}}(\vec{m}) = \{\vec{u} \in U_2 | \mbox{str}_{\alpha_{i_{\alpha}}(\vec{m})}(\vec{u}) > 0 \}$.

We write $\alpha \to^\mathcal{T} \beta$, and we say $\alpha$ \emph{$\mathcal{T}$-produces} $\beta$ (in 0 or more steps) if there is a $\mathcal{T}$-assembly sequence $\alpha_0,\alpha_1,\ldots$ of length $k$ such that
(1) $\alpha = \alpha_0$,
(2) $S_\beta = \bigcup_{0 \leq i < k} S_{\alpha_i}$, and
(3) for all $0 \leq i < k$, $\alpha_{i} \sqsubseteq \beta$.
If $k$ is finite then it is routine to verify that $\beta = \alpha_{k-1}$.

We say $\alpha$ is \emph{$\mathcal{T}$-producible} if $\sigma \to^\mathcal{T} \alpha$, and we write $\prodasm{\mathcal{T}}$ to denote the set of $\mathcal{T}$-producible assemblies. An assembly $\alpha$ is \emph{$\mathcal{T}$-terminal} if $\alpha$ is $\tau$-stable and $\partial^\mathcal{T} \alpha=\emptyset$.
We write $\termasm{\mathcal{T}} \subseteq \prodasm{\mathcal{T}}$ to denote the set of $\mathcal{T}$-producible, $\mathcal{T}$-terminal assemblies. If $|\termasm{\mathcal{T}}| = 1$ then  $\mathcal{T}$ is said to be {\em directed}.

We say that a DTAS $\mathcal{T}$ \emph{strictly (a.k.a. uniquely) self-assembles} a shape $X \subseteq \Z^2$ if, for all $\alpha \in \termasm{\mathcal{T}}$, $S_{\alpha} = X$; i.e., if every terminal assembly produced by $\mathcal{T}$ places tiles on -- and only on -- points in the set $X$.

Now, the DrgTAM is defined a in Section~\ref{sec-rgtam-formal} and a DrgTAS is defined to be a system in the DrgTAM. Note that the glue binding two tiles that form a duple must have strength $1$, and the glues exposed by a duple may have strength $-1$, $0$, or $1$. Also notice that for an assembly $\alpha$ in a DrgTAS, a cut of strength $\leq 0$ may separate two nodes of the grid graph that correspond to two tiles of a duple. Then, the two producible assemblies on each side of this cut each contain one tile from the duple.

} %later

%% file: simulation_def.tex
\vspace{-10pt}
\subsection{Informal Definitions of Simulation}
\label{sec:simulation_def_informal}
\vspace{-5pt}

In this section, we present a high-level sketch of what we mean when saying that one system \emph{simulates} another.  Please see Section~\ref{sec:simulation_def_formal} for complete, technical definitions, which are based on those of \cite{IUNeedsCoop}.

For one system $\mathcal{S}$ to simulate another system $\mathcal{T}$, we allow $\mathcal{S}$ to use square (or rectangular when simulating duples) blocks of tiles called \emph{macrotiles} to represent the simulated tiles from $\mathcal{T}$.  The simulator must provide a scaling factor $c$ which specifies how large each macrotile is, and it must provide a \emph{representation function}, which is a function mapping each macrotile assembled in $\mathcal{S}$ to a tile in $\mathcal{T}$.  Since a macrotile may have to grow to some critical size (e.g. when gathering information from adjacent macrotiles about the simulated glues adjacent to its location) before being able to compute its identity (i.e. which tile from $\mathcal{T}$ it represents), it's possible for non-empty macrotile locations in $\mathcal{S}$ to map to empty locations in $\mathcal{T}$, and we call such growth \emph{fuzz}.  We follow the standard simulation definitions (see \cite{IUNeedsCoop,2HAMIU,Signals3D,IUSA}), and restrict fuzz to being laterally or vertically adjacent to macrotile positions in $\mathcal{S}$ which map to non-empty tiles in $\mathcal{T}$.

Given the notion of block representations, we say that $\mathcal{S}$ simulates $\mathcal{T}$ if and only if (1) for every producible assembly in $\mathcal{T}$, there is an equivalent producible assembly in $\mathcal{S}$ when the representation function is applied, and vice versa (thus we say the systems have \emph{equivalent productions}), and (2) for every assembly sequence in $\mathcal{T}$, the exactly equivalent assembly sequence can be followed in $\mathcal{S}$ (modulo the application of the representation function), and vice versa (thus we say the systems have \emph{equivalent dynamics}).  Thus, equivalent production and equivalent dynamics yield a valid simulation.

\newcommand{\REPL}{\mathsf{REPR}}
\newcommand{\frakC}{\mathfrak{C}}

We say that a tile set $U$ is \emph{intrinsically universal} for a class $\frakC$ of tile assembly systems if, for every system $\calT \in \frakC$ a system $\mathcal{U}_{\mathcal{T}}$ can be created for which: 1. $U$ is the tile set, 2. there is some initial seed assembly consisting of tiles in $U$ which is constructed to encode information about the system $\calT$ being simulated, 3. there exists a representation function $R$ which maps macrotiles in the simulator $\mathcal{U}_{\mathcal{T}}$  to tiles in the simulated system, and 4. under $R$, $\mathcal{U}_{\mathcal{T}}$  has equivalent productions and equivalent dynamics to $\calT$.  Essentially, there is one tile set which can be used to simulate any system in the class, using only custom configured input seed assemblies.  For formal definitions of intrinsic universality in tile assembly, see \cite{IUSA,IUNeedsCoop,Signals3D}.

\ifabstract
\later{
\section{Formal Definitions of Simulation}
\label{sec:simulation_def_formal}

In this section we formally define what it means for an rgTAS to simulate a TAS and a what it means for a DrgTAS to simuate a TAS.

From this point on, let $T$ be a tile set, and let $m\in\Z^+$.
An \emph{$m$-block supertile} or \emph{macrotile} over $T$ is a partial function $\alpha : \Z_m^2 \dashrightarrow T$, where $\Z_m = \{0,1,\ldots,m-1\}$.
Let $B^T_m$ be the set of all $m$-block supertiles over $T$.
The $m$-block with no domain is said to be $\emph{empty}$.
For a general assembly $\alpha:\Z^2 \dashrightarrow T$ and $(x_0, x_1)\in\Z^d$, define $\alpha^m_{x_0,x_1}$ to be the $m$-block supertile defined by $\alpha^m_{x_0, x_1}(i_0, i_1) = \alpha(mx_0+i_0, mx_1+i_1)$ for $0 \leq i_0,i_1< m$.

For some tile set $S$, a partial function $R: B^{S}_m \dashrightarrow T$ is said to be a \emph{valid $m$-block supertile representation} from $S$ to $T$ if for any $\alpha,\beta \in B^{S}_m$ such that $\alpha \sqsubseteq \beta$ and $\alpha \in \dom R$, then $R(\alpha) = R(\beta)$.

For a given valid $m$-block supertile representation function $R$ from tile set~$S$ to tile set $T$, define the \emph{assembly representation function}\footnote{Note that $R^*$ is a total function since every assembly of $S$ represents \emph{some} assembly of~$T$; the functions $R$ and $\alpha$ are partial to allow undefined points to represent empty space.}  $R^*: \mathcal{A}^{S} \rightarrow \mathcal{A}^T$ such that $R^*(\alpha') = \alpha$ if and only if $\alpha(x_0, x_1) = R\left(\alpha'^m_{x_0,x_1}\right)$ for all $(x_0,x_1) \in \Z^2$.
For an assembly $\alpha' \in \mathcal{A}^{S}$ such that $R(\alpha') = \alpha$, $\alpha'$ is said to map \emph{cleanly} to $\alpha \in \mathcal{A}^T$ under $R^*$ if for all non empty blocks $\alpha'^m_{x_0,x_1}$, $(x_0,x_1)+(u_0,u_1) \in \dom \alpha$ for some $u_0,u_1 \in U_2$ such that $u_0^2 + u_1^2 \leq 1$, or if $\alpha'$ has at most one non-empty $m$-block~$\alpha^m_{0, 0}$.  In other words, $\alpha'$ may have tiles on supertile blocks representing empty space in $\alpha$, but only if that position is adjacent to a tile in $\alpha$.  We call such growth ``around the edges'' of $\alpha'$ \emph{fuzz} and thus restrict it to be adjacent to only valid supertiles, but not diagonally adjacent (i.e.\ we do not permit \emph{diagonal fuzz}).

\subsection{rgTAS simulation of a TAS}\label{sec:rgTASsimTAS-formal}

To state our main results, we must formally define what it means for an rgTAS to ``simulate'' a TAS.  Our definitions are similar to the definitions of simulation of a TAS by a TAS given in \cite{IUNeedsCoop}.

In the following definitions, let $\mathcal{T} = \left(T,\sigma_T,\tau_T\right)$ be a TAS, let $\mathcal{S} = \left(S,\sigma_S,\tau_S\right)$ be a TAS, and let $R$ be an $m$-block representation function $R:B^S_m \rightarrow T$.

\begin{definition}
\label{def-equiv-prod} We say that $\mathcal{S}$ and $\mathcal{T}$ have \emph{equivalent productions} (under $R$), and we write $\mathcal{S} \Leftrightarrow \mathcal{T}$ if the following conditions hold:
\begin{enumerate}
        \item $\left\{R^*(\alpha') | \alpha' \in \prodasm{\mathcal{S}}\right\} = \prodasm{\mathcal{T}}$.
        \item For all $\alpha'\in \prodasm{\mathcal{S}}$, $\alpha'$ maps cleanly to $R^*(\alpha')$.
\end{enumerate}
\end{definition}

\begin{definition}
\label{def-t-follows-s} We say that $\mathcal{T}$ \emph{follows} $\mathcal{S}$ (under $R$), and we write $\mathcal{T} \dashv_R \mathcal{S}$ if
(1) $\alpha' \rightarrow_{+}^\mathcal{S} \beta'$, for some $\alpha',\beta' \in \prodasm{\mathcal{S}}$, implies that $R^*(\alpha') \to^\mathcal{T} R^*(\beta')$, and
(2) $\alpha' \rightarrow_{-}^\mathcal{S} (\beta'_1, \beta'_2)$ for some $\alpha',\beta'_1,\beta'_2 \in \prodasm{\mathcal{S}}$, implies that either of the following holds.
\begin{enumerate}
	\item[(a)] $\sigma_S\subseteq \beta'_1$ and $R^*(\alpha') = R^*(\beta'_1)$, and there is some $\gamma \in B^S_m$ such that $\beta'_2 \subseteq \gamma$ and $\gamma \not\in \dom R$, and moreover, if $\beta'_2 \rightarrow^\mathcal{S} \beta''_2$ for some $\beta''_2\in \prodasm{\mathcal{S}}$, then there is some $\gamma' \in B^S_m$ such that $\beta''_2 \subseteq \gamma'$ and $\gamma' \not\in \dom R$.
	\item[(b)] $\sigma_S\subseteq \beta'_2$ and $R^*(\alpha') = R^*(\beta'_2)$, and there is some $\gamma \in B^S_m$ such that $\beta'_1 \subseteq \gamma$ and $\gamma \not\in \dom R$, and moreover, if $\beta'_1 \rightarrow^\mathcal{S} \beta''_1$ for some $\beta''_1\in \prodasm{\mathcal{S}}$, then there is some $\gamma' \in B^S_m$ such that $\beta''_1 \subseteq \gamma'$ and $\gamma' \not\in \dom R$.
\end{enumerate}
\end{definition}

Condition (2) in the definition above says that when a cut is made to an assembly $\alpha'\in\prodasm{S}$ that represents $\alpha\in\prodasm{T}$, the two assemblies that are produced are such that one of the assemblies, $\beta'_1$ say, still represents $\alpha$ and is identifiable by the fact that it contains the seed $\sigma_S$, while the other assembly, $\beta'_2$, represents the empty tile. In addition, the result of any assembly sequence starting from $\beta'_2$ must also represent the empty tile. Informally, ``junk'' that falls off of an assembly during simulation must represent the empty tile and cannot grow into anything other than an assembly that represents the empty tile.

\begin{definition}
\label{def-s-models-t} We say that $\mathcal{S}$ \emph{models} $\mathcal{T}$ (under $R$), and we write $\mathcal{S} \models_R \mathcal{T}$, if for every $\alpha \in \prodasm{\mathcal{T}}$, there exists $\Pi \subset \prodasm{\mathcal{S}}$ where $R^*(\alpha') = \alpha$ for all $\alpha' \in \Pi$, such that, for every $\beta \in \prodasm{\mathcal{T}}$ where $\alpha \rightarrow^\mathcal{T} \beta$, (1) for every $\alpha' \in \Pi$ there exists $\beta' \in \prodasm{\mathcal{S}}$ where $R^*(\beta') = \beta$ and $\alpha' \rightarrow^\mathcal{S} \beta'$, and (2) for every $\alpha'' \in \prodasm{\mathcal{S}}$ where $\alpha'' \rightarrow^\mathcal{S} \beta'$, $\beta' \in \prodasm{\mathcal{S}}$, $R^*(\alpha'') = \alpha$, and $R^*(\beta') = \beta$, there exists $\alpha' \in \Pi$ such that $\alpha' \rightarrow^\mathcal{S} \alpha''$.
\end{definition}

The previous definition essentially specifies that every time $\mathcal{S}$ simulates an assembly $\alpha \in \prodasm{\mathcal{T}}$, there must be at least one valid growth path in $\mathcal{S}$ for each of the possible next steps that $\mathcal{T}$ could make from $\alpha$ which results in an assembly in $\mathcal{S}$ that maps to that next step.

\begin{definition}
\label{def-s-simulates-t} We say that $\mathcal{S}$ \emph{simulates} $\mathcal{T}$ (under $R$) if $\mathcal{S} \Leftrightarrow_R \mathcal{T}$ (equivalent productions), $\mathcal{T} \dashv_R \mathcal{S}$ and $\mathcal{S} \models_R \mathcal{T}$ (equivalent dynamics).
\end{definition}

\subsection{Dupled rgTAS simulation of a TAS}\label{sec:DrgTASsimTAS-formal}

Here we formally define what it means for a DrgTAS to ``simulate'' a TAS.  The definition of a DrgTAS lends itself to a simulation definition statement that is equivalent to the definition of simulation for a TAS simulating another TAS. Therefore, our definitions come from \cite{IUNeedsCoop}.

Now let $\mathcal{T} = \left(T,\sigma_T,\tau_T\right)$ be a TAS, let $\mathcal{U} = \left(U,S,D,\sigma_U\right)$ be a DrgTAS, and let $R$ be an $m$-block representation function $R:B^U_m \rightarrow T$. Then we may define \emph{equivalent production}, \emph{follows}, and \emph{models} for $\mathcal{U}$ and $\mathcal{T}$ (under $R$) exactly as defined in Section~\ref{sec:rgTASsimTAS-formal} and therefore define simulation as follows.

\begin{definition}
\label{def-d-simulates-t} We say that $\mathcal{U}$ \emph{simulates} $\mathcal{T}$ (under $R$) if $\mathcal{U} \Leftrightarrow_R \mathcal{T}$ (equivalent productions), $\mathcal{T} \dashv_R \mathcal{U}$ and $\mathcal{U} \models_R \mathcal{T}$ (equivalent dynamics).
\end{definition}
} %later

%% file: rgTAS_cannot_coop.tex
\vspace{-10pt}
\section{A temperature-$2$ aTAM system that cannot be simulated by any rgTAS}\label{rgTAS_cannot_coop}
\vspace{-10pt}

In this section we show that there exists a temperature-$2$ aTAM system that cannot be simulated by any rgTAM system. Here we give an overview of the TAS, $\mathcal{T}$, that we show cannot be simulated by any rgTAS, and an overview of the proof. For details of the proof, see Section~\ref{sec:negResProof} in the Appendix.

\vspace{-5pt}
\begin{theorem}\label{thm:rgTAScannotSIMaTAM}
There exists a temperature-$2$ aTAM system $\mathcal{T} = (T,\sigma, 2)$ such that $\mathcal{T}$ cannot be simulated by any rgTAS.
\end{theorem}
\vspace{-5pt}

\begin{figure}[htb]
\begin{center}
\includegraphics[width=3.5in]{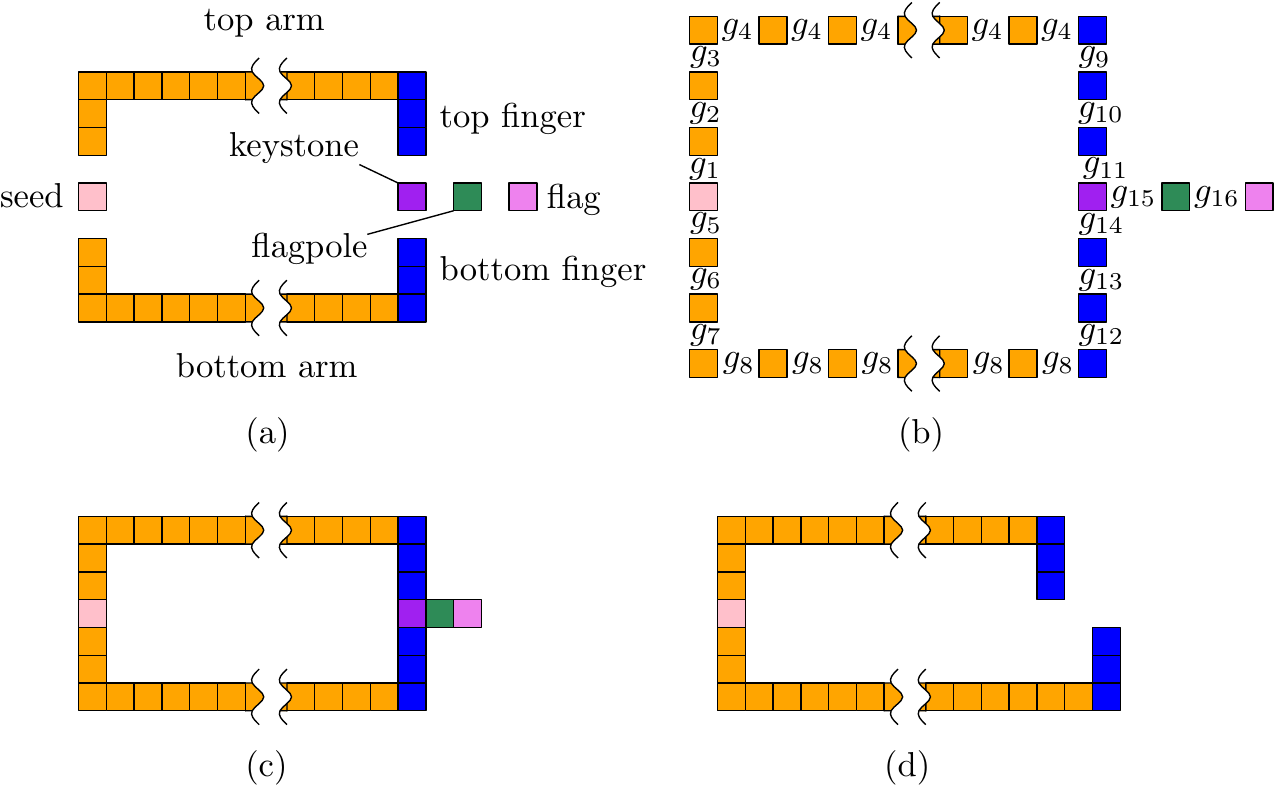}
\caption{(Figure taken from \cite{IUNeedsCoop}) (a) An overview of the tile assembly system $\mathcal{T} = (T,\sigma,2)$.~$\mathcal{T}$ runs at temperature 2 and its tile set $T$ consists of 18 tiles. (b) The glues used in the tileset $T$. Glues $g_{11}$ and $g_{14}$ are strength 1, all other glues are strength~2.  Thus the keystone tile binds with two ``cooperative'' strength~1 glues. Growth begins from the pink seed tile $\sigma$: the top and bottom arms are one tile wide and grow to arbitrary, nondeterministically chosen, lengths. Two blue figures grow as shown. (c) If the fingers happen to meet then the keystone, flagpole and flag tiles are placed, (d) if the fingers do not meet then growth terminates at the finger ``tips''.}
\label{fig:fingerFlagpole_overview}
\end{center}
\vspace{-25pt}
\end{figure}

Let $\mathcal{T} = (T, \sigma, 2)$ denote the system with $T$ and $\sigma$ given in Figure~\ref{fig:fingerFlagpole_overview}. The glues in the various tiles are all unique with the exception of the common east-west glue type used within each arm to induce non-deterministic and independent arm lengths. Glues are shown in part (b) of Figure~\ref{fig:fingerFlagpole_overview}.
Note that cooperative binding happens at most once during growth, when attaching the keystone tile to two arms of identical length. All other binding events are noncooperative and all glues are strength $2$ except for $g_{11}, g_{14}$ which are strength $1$.

The TAS $\mathcal{T}$ was used in~\cite{IUNeedsCoop} to show that there is a temperature-$2$ aTAM system that cannot be simulated by a temperature-$1$ aTAM system. To prove that there is no rgTAS that simulates $\mathcal{T}$, we use a similar proof to the proof for aTAM systems, however, we must take special care to show that allowing for a single negative glue does not give enough strength to the model to allow for simulation of cooperative glue binding.

The proof is by contradiction. Suppose that $\mathcal{S} = (S,\sigma_S)$ is an rgTAS that simulates $\mathcal{T}$. We call an assembly sequence $\vec{\alpha} = (\alpha_0, \alpha_1, \dots)$ in an rgTAS \emph{detachment free} if for all $i\geq0$, $\alpha_{i+1}$ is obtained from $\alpha_i$ by the stable attachment of a single tile. The following lemma gives sufficient conditions for the existence of a detachment free assembly sequence.

\vspace{-5pt}
\begin{lemma}\label{lem:stable_assembly-main}
Let $\mathcal{S} = (S, \sigma_S)$ be an rgTAS and let $\alpha\in \prodasm{S}$ be a finite stable assembly. Furthermore, let $\beta$ be a stable subassembly of $\alpha$. Then there exists a detachment free assembly sequence $\vec{\alpha} = (\alpha_1, \alpha_2, \dots, \alpha_{n})$ such that $\alpha_1 = \beta$, and $\alpha_n=\alpha$.
\end{lemma}
\vspace{-10pt}

\begin{figure}[htp]
\begin{center}
\includegraphics[width=4.7in]{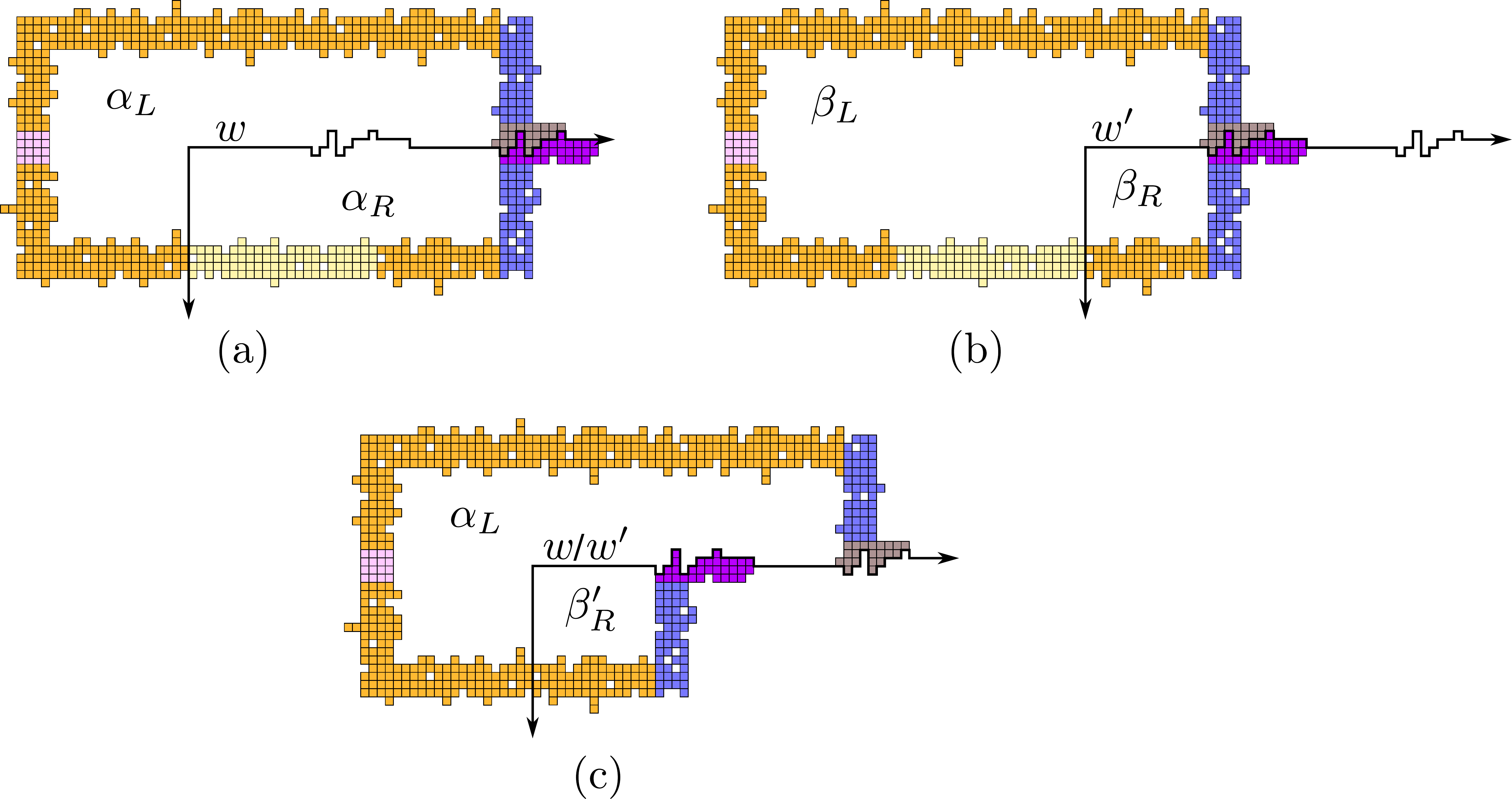}
\caption{An example assembly formed by $S$ simulating $\mathcal{T}$ -- (a) and (b), and the resulting producible assembly (c) constructed via a ``splicing'' technique that uses the window movie lemma. The assembly in (c) shows that $\mathcal{S}$ is incapable of valid simulation of $\mathcal{T}$.}
\label{fig:bad_sim_overview-main}
\end{center}
\vspace{-25pt}
\end{figure}

A corollary of this lemma is that if an rgTAS gives a valid simulation of $\mathcal{T}$, it can do so using detachment free assembly sequences. Using detachment free assembly sequences, it is possible to use a technique for ``splicing'' subassemblies of producible assemblies of $\mathcal{S}$.
This technique uses a lemma referred to as the ``window movie lemma''. For aTAM systems, this lemma is shown in~\cite{IUNeedsCoop} (Lemma 3.1). We give a version of the window movie lemma that holds for detachment free assembly sequences. See Section~\ref{sec:negResProof} for the formal definitions of windows and window movies, and for a formal statement of the window movie lemma that we use.  Figure~\ref{fig:bad_sim_overview-main} gives a depiction of this splicing technique. Here we use this lemma for detachment free assembly in the rgTAM. Then, using this splicing technique, we show that if $\mathcal{S}$ can simulate $\mathcal{T}$, it can also produce assemblies that violate the definition of simulation. In other words, we arrive at our contradiction and conclude that there is no rgTAS that can simulate $\mathcal{T}$.

\ifabstract
\later{

\section{Proof of Theorem~\ref{thm:rgTAScannotSIMaTAM}}\label{sec:negResProof}

Before we prove Theorem~\ref{thm:rgTAScannotSIMaTAM} we will give necessary conditions for any rgTAS system that can simulate $\mathcal{T}$. Let $\mathcal{S} = (S,\sigma_S)$ denote any rgTAS that simulates $\mathcal{T}$. We call an assembly sequence $\vec{\alpha} = (\alpha_0, \alpha_1, \dots)$ in an rgTAS \emph{detachment free} if for all $i\geq0$, $\alpha_{i+1}$ is obtained from $\alpha_i$ by the stable attachment of a single tile. The following lemma gives sufficient conditions for the existence of a detachment free assembly sequence.

\begin{lemma}\label{lem:stable_assembly}
Let $\mathcal{S} = (S, \sigma_S)$ be an rgTAS and let $\alpha\in \prodasm{S}$ be a finite stable assembly. Furthermore, let $\beta$ be a stable subassembly of $\alpha$. Then there exists a detachment free assembly sequence $\vec{\alpha} = (\alpha_1, \alpha_2, \dots, \alpha_{n})$ such that $\alpha_1 = \beta$, and $\alpha_n=\alpha$.
\end{lemma}

\begin{proof}
Let $W$ be the set of of subassemblies of $\alpha$ such that $\eta\in W$ if and only if there exists an assembly sequence consisting of stable assemblies starting from $\beta$ with result $\eta$ that is detachment free. Note that since $\alpha$ is finite, $W$ is finite. Therefore, we can let $\gamma$ denote a subassembly of $W$ such that for any $\eta$ in $W$, $|\dom \gamma| \geq |\dom \eta|$. In other words, $\gamma$ is such that no other subassembly in $W$ has more tiles than $\gamma$. We will show that $\gamma = \alpha$.

For the sake of contradiction, assume that $\gamma \neq \alpha$. Then there is some tile of $\alpha$ that is not in $\gamma$. Consider the binding graph of $\alpha$ with nodes corresponding to tiles of $\gamma$ removed, and call the resulting graph $G$. Notice that a connected component (possibly with edges corresponding to the negative glue) of $G$ corresponds to a subassembly of tiles, $x$ say, in $\alpha$ such that no tile of $x$ is in $\gamma$. Now, since $\alpha$ is stable, the cut $c$ of the binding graph of $\alpha$ that separates $x$ from $\alpha$ must have strength greater than $0$. Since $x$ is taken to be a connected component of $G$, all of the edges defining the cut $c$ correspond to exposed glues of $\gamma$. Since the strength of these edges sum to a positive strength, at least one tile of $x$ can stably bind to $\gamma$ resulting in $\gamma'$ because at least on position must receive positive strength across the cut.  Note that $\gamma'$ is in $W$ since it is obtained from $\gamma$ by a single tile addition. Finally, the fact that $|\dom \gamma'| = |\dom \gamma| + 1 > |\dom \gamma|$, contradicts our choice of $\gamma$.

\end{proof}

The following lemma states that if an rgTAS gives a valid simulation of $\mathcal{T}$, it can do so using detachment free assembly sequences.

\begin{corollary}\label{cor:detachmentfree}
Let $\mathcal{S} = (S, \sigma_S)$ be an rgTAS that simulates $\mathcal{T}$ under $R$, and let $\alpha$ be in $\prodasm{T}$. Then there exists a stable assembly $\alpha'' \in \prodasm{S}$ and a detachment free assembly sequence $\vec{\alpha}$ starting from $\sigma_S$ with result $\alpha''$ such that $\alpha''$ represents $\alpha$ under $R$.
\end{corollary}

\begin{proof}
Let $\alpha'$ be in $\prodasm{S}$ such that $\alpha'$ represents $\alpha$ under $R$.
We obtain $\alpha''$ from $\alpha'$ by allowing detachment to occur for each cut of $\alpha'$ with strength $<1$. In particular, there exists an assembly sequence $\vec{\alpha}_{d} = (\alpha_1, \alpha_2, \dots, \alpha_n)$ where $\alpha_1 = \alpha'$, $\alpha_n = \alpha''$, and $\alpha_{i+1}$ is obtained from $\alpha_i$ by the detachment along a strength $<1$ cut. The existence of $\vec{\alpha}_d$ follows from the fact that as detachment occurs in $\alpha_i$ along a cut $c$, one side of the cut must be an assembly that maps to $\alpha$ under $R$ (by the definition of simulation in Section~\ref{sec:simulation_def_formal}). We take this assembly to be $\alpha_{i+1}$.
Therefore, we have a stable assembly $\alpha''$ that represents $\alpha$ under $R$. Finally, since the seed $\sigma_S$ is a stable subassembly of $\alpha''$, by Lemma~\ref{lem:stable_assembly} there exists a detachment free assembly sequence $\vec{\alpha}$ with result $\alpha''$.
\end{proof}

To show that $\mathcal{T}$ cannot be simulated by an rgTAS, we will use the window movie lemma. This lemma was introduced in~\cite{IUNeedsCoop} (Lemma 3.1) and was used to show that there does not exist a temperature $1$ aTAM system that can simulated $\mathcal{T}$. We will start by stating the definitions of a window and window movie.

\begin{definition}
A \emph{window} $w$ is a set of edges forming a cut-set in the infinite grid graph.
\end{definition}

Often a window is depicted as paths (possibly closed) in the 2D plane. See Figure~\ref{fig:bad_sim_overview} for an example. Given a window and an assembly sequence, one can observe the order and sequence that tiles attach across the window. This gives rise to the following definition.

\begin{definition}\label{def:windowMovie}
Given an assembly sequence $\vec{\alpha}$ and a window $w$, the associated {\em window movie} is the maximal sequence $M_{\vec{\alpha},w} = (v_{0}, g_{0}) , (v_{1}, g_{1}), (v_{2}, g_{2}), \ldots$ of pairs of grid graph vertices $v_i$ and glues $g_i$, given by the order of the appearance of the glues along window $w$ in the assembly sequence $\vec{\alpha}$.
Furthermore, if $k$ glues appear along $w$ at the same instant (this happens upon placement of a tile which has multiple  sides  touching $w$) then these $k$ glues appear contiguously and are listed in lexicographical order of the unit vectors describing their orientation in $M_{\vec{\alpha},w}$.
\end{definition}

Now we can state the window movie lemma for detachment free assembly sequences.

\begin{lemma}[Window movie lemma]
\label{lem:windowmovie}
Let $\vec{\alpha} = (\alpha_i \mid 0 \leq i < l)$ and $\vec{\beta} = (\beta_i \mid 0 \leq i < m)$, with
$l,m\in\Z^+ \cup \{\infty\}$,
be \emph{detachment free} assembly sequences in $\mathcal{T}$ with results $\alpha$ and $\beta$, respectively.
Let $w$ be a window that partitions~$\alpha$ into two configurations~$\alpha_L$ and $\alpha_R$, and $w' = w + \vec{c}$ be a translation of $w$ that partitions~$\beta$ into two configurations $\beta_L$ and $\beta_R$.
Furthermore, define $M_{\vec{\alpha},w}$, $M_{\vec{\beta},w'}$ to be the respective window movies for $\vec{\alpha},w$ and $\vec{\beta},w'$, and define $\alpha_L$, $\beta_L$ to be the subconfigurations of $\alpha$ and $\beta$ containing the seed tiles of $\alpha$ and $\beta$, respectively.
Then if $M_{\vec{\alpha},w} = M_{\vec{\beta},w'}$, it is the case that  the following two assemblies are also producible:
(1) the assembly $\alpha_L \beta'_R = \alpha_L \cup \beta'_R$ and
(2) the assembly $\beta'_L \alpha_R = \beta'_L \cup \alpha_R$, where $\beta'_L=\beta_L-\vec{c}$ and $\beta'_R=\beta_R-\vec{c}$.
\end{lemma}

Under the assumption that the assembly sequences in Lemma~\ref{lem:windowmovie} are detachment free, Lemma~\ref{lem:windowmovie} follows directly from the proof of the window movie lemma for aTAM systems (Lemma 3.1 in~\cite{IUNeedsCoop}).
We can also define a restricted form of a window movie. For windows $w$ and $w'$, and assembly sequences $\vec{\alpha}$ and $\vec{\beta}$, Lemma~\ref{lem:windowmovie} holds even if the window movies
$M_{\vec{\alpha},w}$ and $M_{\vec{\alpha},w'}$ match on specific \emph{submovies} (subsequences of the movies $M_{\vec{\alpha},w}$ and $M_{\vec{\alpha},w'}$). We specify a particular submovie as follows.

Consider the window movie $M_{\vec{\alpha},w}$. Location-glue pairs are added to a window movie by observing tile placements given by $\vec{\alpha}$. Suppose that step $i$ of $\vec{\alpha}$ is the placement of a tile $t$ that adds a location-glue pair $(l,g)$ to the window movie.  We call this tile placement \emph{non-window crossing} if the tile can stably bind even in the absence of any positive glue along the window $w$.
We also define a \emph{window crossing submovie} to be the subsequence of a window movie, $M$, that consists of all of the steps of $M$ except for the steps corresponding to the addition of a non-window crossing tile. We denote the window crossing submovie of $M$ by ${\cal W}(M)$. Note that every window movie has a unique window crossing submovie. Then, Corollary~\ref{cor:windowmovie} says that in certain cases, Lemma~\ref{lem:windowmovie} holds even if two window movies only match on their window crossing submovies.

\begin{corollary}
\label{cor:windowmovie}
Suppose that the following two conditions hold.
\begin{enumerate}
	\item[(1)] For all $(l,g)$ in $M_{\vec{\alpha},w}$ such that $(l,g)$ corresponds to the placement of a tile $t$ with north glue $g$,
if there exists a tile $t'$ in $\beta$ at location $l' = l + c + (0,1)$ such that the south glue $g'$ of $t$ and $g$ are the negative glue, then there exists a tile in $\alpha$ at location $l + (0,1)$ with south glue $g$. We also include the similar conditions for $(l,g)$ in $M_{\vec{\alpha},w}$ where $g$ is a south, east, or west glue.
	\item[(2)] For all $(l',g')$ in $M_{\vec{\beta},w'}$ such that $(l',g')$ corresponds to the placement of a tile $t'$ with north glue $g'$,
if there exists a tile $t$ in $\alpha$ at location $l = l' - c + (0,1)$ such that the south glue $g$ of $t$ and $g'$ are the negative glue, then there exists a tile in $\beta$ at location $l + (0,1)$ with south glue $g'$. We also include the similar conditions for $(l',g')$ in $M_{\vec{\beta},w'}$ where $'g$ is a south, east, or west glue.
\end{enumerate}
Then, the statement of Lemma~\ref{lem:windowmovie} holds if the window movies $M_{\vec{\alpha},w}$ and $M_{\vec{\beta},w'}$ are replaced by their window crossing submovies ${\cal W}\left(M_{\vec{\alpha},w}\right)$ and ${\cal W}\left(M_{\vec{\beta},w'}\right)$.
\end{corollary}

\begin{figure}[htb]
\begin{center}
\includegraphics[width=4.5in]{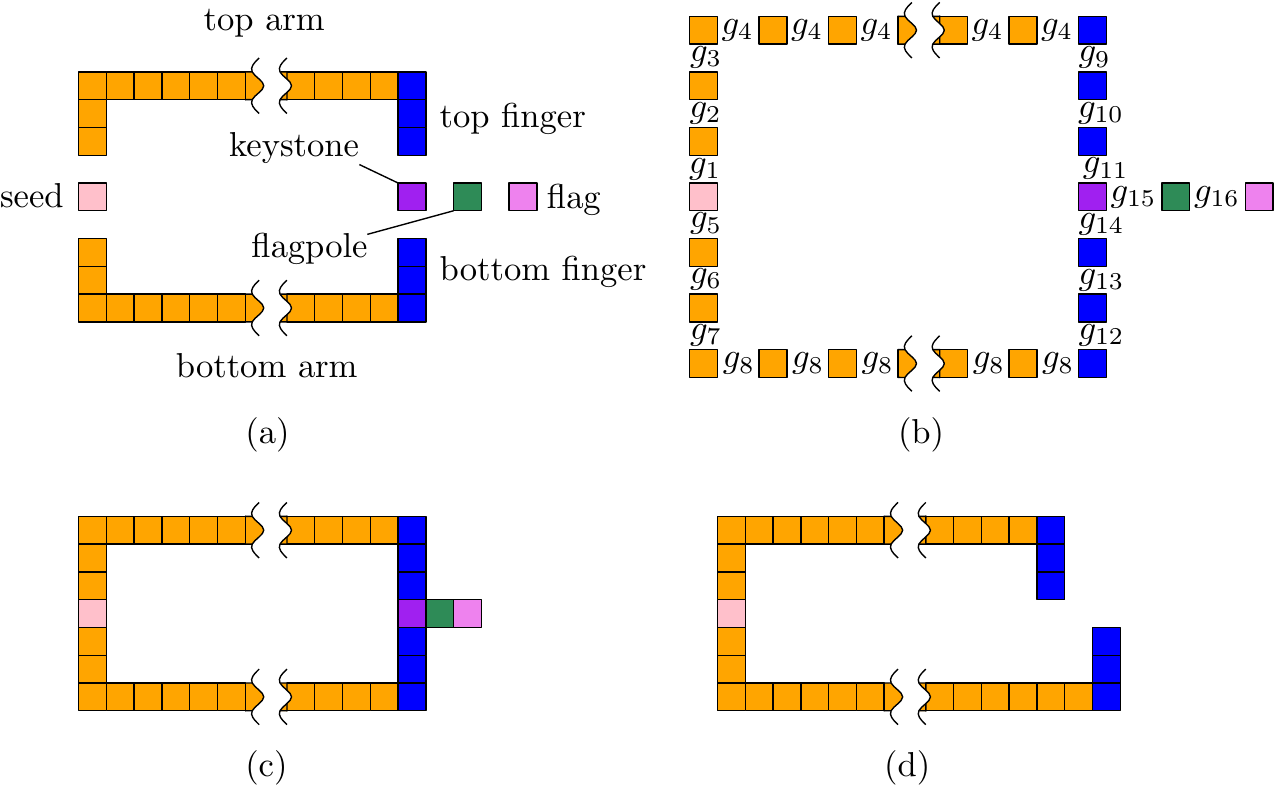}
\caption{(Figure taken from \cite{IUNeedsCoop}) (a) An overview of the tile assembly system $\mathcal{T} = (T,\sigma,2)$.~$\mathcal{T}$ runs at temperature 2 and its tile set $T$ consists of 18 tiles. (b) The glues used in the tileset $T$. Glues $g_{11}$ and $g_{14}$ are strength 1, all other glues are strength~2.  Thus the keystone tile binds with two ``cooperative'' strength~1 glues. Growth begins from the pink seed tile $\sigma$: the top and bottom arms are one tile wide and grow to arbitrary, nondeterministically chosen, lengths. Two blue figures grow as shown. (c) If the fingers happen to meet then the keystone, flagpole and flag tiles are placed, (d) if the fingers do not meet then growth terminates at the finger ``tips''.}
\label{fig:fingerFlagpole_overview_append}
\end{center}
\end{figure}

Condition (1) in Corollary~\ref{cor:windowmovie} is saying that when we attempt to assemble $\alpha_L \beta'_R$, we can rest assured that there are no negative glue interactions across the window $w$ between negative glues exposed by tiles of $\beta'_R$ and negative glues exposed by $\alpha_L$ that are not present in the assembly $\alpha$. This implies that using the assembly sequence $\vec{\alpha}$ to attach tiles from $\alpha_L$, and the assembly sequence $\vec{\beta}$ to attach tiles from $\beta'_R$, $\alpha_L \beta'_R$ can be assembled since there are no negative glue interactions in $\alpha_L \beta'_R$ that are not present in $\alpha$ or $\beta$. Similarly, Condition (2) says the same for the assembly of $\beta'_L \alpha_R$.

\begin{figure}[htp]
\begin{center}
\includegraphics[width=4.5in]{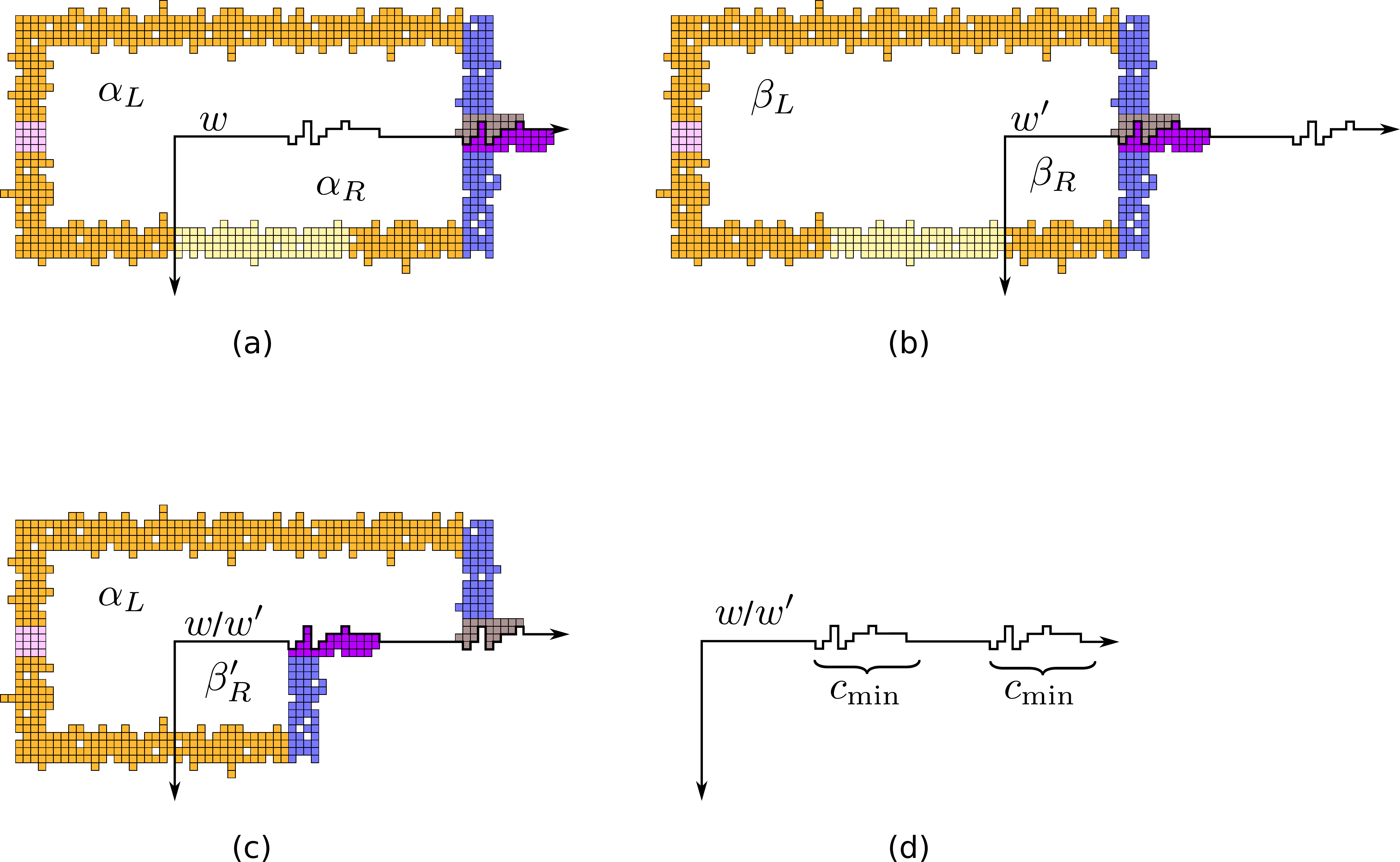}
\caption{An example of an assembly formed by $S$ simulating $\mathcal{T}$ and the identical window crossing submovies $w$ and $w'$ -- (a) and (b), and the resulting producible assembly constructed via Corollary~\ref{cor:windowmovie} (c). (d) shows the windows $w$ and $w'$ (which are equivalent up to shifting). The portions of these windows that are determined by $c_{\min}$ are labeled.}
\label{fig:bad_sim_overview}
\end{center}
\end{figure}

With Corollary~\ref{cor:detachmentfree} and Corollary~\ref{cor:windowmovie}, we are now ready to prove Theorem~\ref{thm:rgTAScannotSIMaTAM}. For the sake of contradiction, suppose that $\mathcal{S} = (S, \sigma_S)$ is an rgTAS that simulates $\mathcal{T}$, the finger and flagpole system, $\mathcal{T}$ with representation function $R: \mathcal{A}^{S} \rightarrow \mathcal{A}^T$ and scale factor $m\in \mathbb{N}$.

Now let $\alpha_d$ in $\termasm{T}$ be the assembly where the top and bottom arms are $d$ tiles long.
By Corollary~\ref{cor:detachmentfree}, we can find a detachment free assembly sequence $\vec{\alpha}'_d$ in $\mathcal{S}$ such that the stable result $\alpha'_d$ represents $\alpha_d$.
Now let $c$ be a set of edges in the binding graph $G$ of $\alpha'_d$ such that $c$ is a cut-set of the subgraph of $G$ corresponding to the subassembly, $\eta$, of tiles contained in the keystone macrotile, the flagpole macrotile, the flag macrotile, and the macrotiles immediately surrounding these macrotiles in $\alpha'_d$. Then let $C$ be the set of all such cuts $c$. Since $|C| < \infty$, we can find a cut $c_{\min}$ such that for any cut $c$ in $C$, the strength of $c_{\min}$ is less than or equal to the strength of $c$. In other words, $c_{\min}$ is a cut with minimal strength.

For the proof here, we must be more selective about our choice of assembly sequence $\vec{\alpha}'_d$ resulting in $\alpha'_d$. In this proof, we will use the window movie lemma for detachment free assembly sequences (Lemma~\ref{lem:windowmovie}). For some $d$ to be chosen later, the windows, $w$ and $w'$, that we will use for Lemma~\ref{lem:windowmovie} will be windows that cut an arm of $\alpha'_d$ vertically.
Note that we can also ensure that other than the edges corresponding to bonds between tiles of belonging to macrotiles of an arm, the only edges in $w$ or $w'$ are exactly the edges of $c_{\min}$. Moreover, without loss of generality, suppose that a tile in the flagpole region stably binds below the cut $c_{\min}$. We will choose the windows $w$ and $w'$ to cut the bottom arm of $\alpha_d'$. See Figure~\ref{fig:bad_sim_overview} for an example of such windows.

\begin{claim}
$\vec{\alpha}'_d$ can be chosen so that every location-glue pair of $M_{\vec{\alpha}'_d, w}$ or $M_{\vec{\alpha}'_d, w'}$ whose glues lie on $c_{\min}$ corresponds to a tile placement that is non-window crossing.
\end{claim}

For the moment, suppose that the claim holds and $\vec{\alpha}'_d$ is chosen as such. Then, let $g$ be the number of glues of tiles in $S$. We will show that $\mathcal{S}$ is capable of producing an assembly sequence that yields an invalid production for simulation. For any $d \in \mathbb{N}$, it must be the case that $\mathcal{S}$ can simulate the production of the assembly $\alpha_d$ in $\termasm{T}$ where the top and bottom arms of $\alpha_d$ are $d$ tiles long. Note that for every $d$, $\alpha_d$ is of the form depicted (c) of Figure~\ref{fig:fingerFlagpole_overview_append}.
Figure~\ref{fig:bad_sim_overview} shows our choice of windows $w$ and $w'$ that cut an arm of some $\alpha_d'$ vertically. By the claim, we can assume that every location-glue pair of $M_{\vec{\alpha}'_d, w}$ and $M_{\vec{\alpha}'_d, w'}$ corresponds to non-window crossing tile additions forming $\eta$. Therefore, the window crossing submovies $\mathcal{W}(M_{\vec{\alpha}'_d, w})$ and $\mathcal{W}(M_{\vec{\alpha}'_d, w'})$ only contain location-glue pairs corresponding to the bindings of tiles belonging to the bottom arm of $\alpha_d'$ (i.e. location-glue pairs along the vertical portion of the windows).

Then, since $m$ (macrotile size) and $g$ (the number of glues of tile types in $S$) are fixed constants, for $d$ sufficiently large, there exists two such window movies $w$ and $w'$ such that $w'$ is a horizontal translation $w$ and the window crossing submovies, $\mathcal{W}(M_{\vec{\alpha}'_d, w})$ and $\mathcal{W}(M_{\vec{\alpha}'_d, w'})$ match.  The top assemblies in Figure~\ref{fig:bad_sim_overview} give an example of two equivalent window movies. Notice that we can also choose $w$ and $w'$ so that the distance between them is at least $3m$. Then, $w$ (respectively $w'$) divides $\alpha_d'$ into configurations $\alpha_L$ and $\alpha_R$ (respectively $\beta_L$ and $\beta_R$). By Corollary~\ref{cor:windowmovie}, $\alpha_L\beta'_R$ (depicted in Figure~\ref{fig:bad_sim_overview}(c)) is a valid assembly in $\mathcal{S}$. Notice that $\alpha_L\beta'_R$ is stable and contains a tile in the flagpole macrotile region. This region lies outside of any permissible fuzz region. (See Section~\ref{sec:simulation_def_formal} for the definition of fuzz.) Therefore, the existence of the valid producible assembly $\alpha_L\beta'_R$ shows that $\mathcal{S}$ is not a valid simulation.

To finish the proof, we now prove the claim.

\noindent\textit{Proof of the claim.} Here we show that $\vec{\alpha}'_d$ as defined above can be chosen so that each glue lying on $c_{\min}$ corresponds to a tile placement that is non-window crossing.
The proof of this claim is similar to the proof of Lemma~\ref{lem:stable_assembly}.

First, let $W$ be the set of of subassemblies of $\alpha_d'$ such that $\eta\in W$ if and only if there exists an assembly sequence consisting of stable assemblies starting from $\sigma_S$ with result $\eta$ that is detachment free \emph{and} every location-glue pair of $M_{\vec{\alpha}'_d, w}$ (The proof is similar for $M_{\vec{\alpha}'_d, w'}$.)
Note that since $\alpha_d'$ is finite, $W$ is finite. Therefore, we can let $\gamma$ denote a subassembly of $W$ such that for any $\eta$ in $W$, $|\dom \gamma| \geq |\dom \eta|$. In other words, $\gamma$ is such that no other subassembly in $W$ has more tiles than $\gamma$. We will show that $\gamma = \alpha_d'$.

For the sake of contradiction, assume that $\gamma \neq \alpha_d'$. Then there is some tile of $\alpha_d'$ that is not in $\gamma$. Consider the binding graph of $\alpha_d'$ with nodes corresponding to tiles of $\gamma$ removed, and call the resulting graph $G$. Notice that a connected component (possibly with edges corresponding to the negative glue) of $G$ corresponds to a configuration of tiles, $x$ say, in $\alpha_d'$ such that no tile of $x$ is in $\gamma$. Now, since $\alpha_d'$ is stable, the cut $c$ of the binding graph of $\alpha_d'$ that separates $x$ from $\alpha_d'$ must have strength greater than $0$. Since $x$ is taken to be a connected component of $G$, all of the edges defining the cut $c$ correspond to exposed glues of $\gamma$. Since the strength of these edges sum to a positive strength, either (1) at least one tile of $x$ can stably bind to $\gamma$ resulting in $\gamma'$ in $W$, or (2) no tile can stably bind to $\gamma$ without the added strength of binding to a glue corresponding to an edge of $c_{\min}$.

\begin{figure}[htp]
\begin{center}
\includegraphics[width=4.5in]{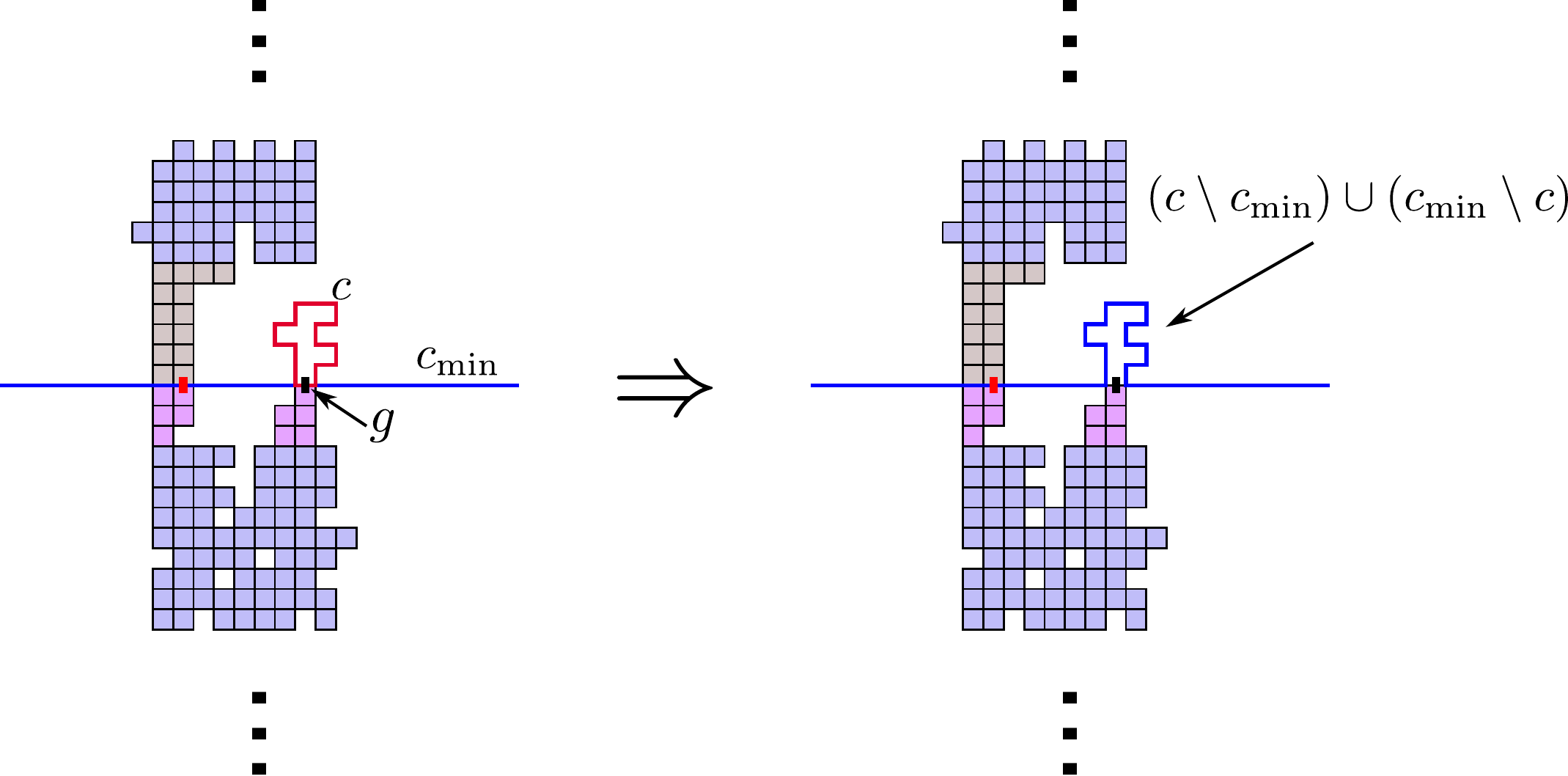}
\caption{A schematic picture of ``rewiring'' the cut $c_{\min}$ of $\eta$. On the right we see the cut $c$ as well as $c_{\min}$. $c_{\min}$ is a cut of strength $0$, and the only positive strength glue on the cut $c$ is labeled $g$ in the figure. On the right, we see $(c \setminus c_{\min}) \cup (c_{\min} \setminus c)$. Notice that this new cut has strength less than the strength of $c_{\min}$.}
\label{fig:rewire_cut}
\end{center}
\end{figure}

In Case (1), note that $|\dom \gamma'| = |\dom \gamma| + 1 > |\dom \gamma|$. This contradicts our choice of $\gamma$. In Case (2), it must be the case that the cut $c$ and the cut $c_{\min}$ share some edges with positive strength. This is because the reason we cannot place a tile using a positive strength glue on $c$ is that this glue is also in $c_{\min}$, and we are not allowing tile attachment of window crossing tiles across $c_{\min}$ in the assembly of $\gamma$. Then, notice that the sum of the strengths of the edges belonging to $c \setminus c_{\min}$ must sum to zero or less. Otherwise a tile could be added along this cut, which would once again contradict our choice of $\gamma$. Then, note that the edges in $(c \setminus c_{\min}) \cup (c_{\min} \setminus c)$ form a cut of the subassembly $\eta$ (defined above the statement of the claim) with strength strictly less than the strength of $c_{\min}$. Intuitively, $(c \setminus c_{\min}) \cup (c_{\min} \setminus c)$ is a cut that is formed by ``rewiring'' $c_{\min}$ using  $c \setminus c_{\min}$, and since the strength of $c \setminus c_{\min}$ is less than $1$, and the strength of $c \cap c_{\min}$ is greater than $0$, this rewiring results in cut with less strength than $c_{\min}$. See Figure~\ref{fig:rewire_cut} for a schematic picture of this rewiring. This contradicts our choice of $c_{\min}$. Hence, in either Case (1) or (2), we arrive at a contradiction. Therefore, $\gamma = \alpha_d'$. This proves the claim.

} %later 

%% file: simCoop.tex
\vspace{-10pt}
\section{Simulation of the aTAM with the DrgTAM}\label{sec:DrgTAM_sim}
\vspace{-10pt}
In this section, given an aTAM system $\calT=(T,\sigma, 2)$, we describe how to simulate $\calT$ with a DrgTAS at temperature 1 with $O(1)$ scale factor and tile complexity $O(|T|)$.  It will then follow from \cite{IUSA} that there exists a tile set in the DrgTAM at $\tau=1$ which is intrinsically universal for the aTAM at any temperature, i.e. it can be used to simulate any aTAM system of any temperature.

\vspace{-5pt}
\begin{theorem}\label{thm:DrgTAS-sim}
For every aTAM system $\calT=(T,\sigma, 2)$, there exists a DrgTAS $\mathcal{D} = (T_{\mathcal{D}}, S, D, \sigma', 1)$ such that $\mathcal{D}$ simulates $\calT$ with $O(1)$ scale factor and $|S\cup D| = O(|T|)$.
\end{theorem}
\vspace{-5pt}
%\subsection{Construction Overview}
We now provide a high-level overview of the construction.  For the remainder of this section, $\calT=(T,\sigma, 2)$ will denote an arbitrary TAS being simulated, $\mathcal{D} = (T_{\mathcal{D}}, S, D, \sigma', 1)$ the simulating DrgTAS, and $R$ the representation function which maps blocks of tiles in $\mathcal{D}$ to tiles in $\calT$.  The system $\calT$ is simulated by a DrgTAS through the use of macrotiles which consist of the components shown in Figure~\ref{fig:macro_labeled-main}.  Note that macrotiles are not necessarily composed of all of the components shown in Figure~\ref{fig:macro_labeled-main}, but will consist of at least one of the subassemblies labeled probe.
\begin{figure}[htp]
\begin{center}
\includegraphics[width=3.5in]{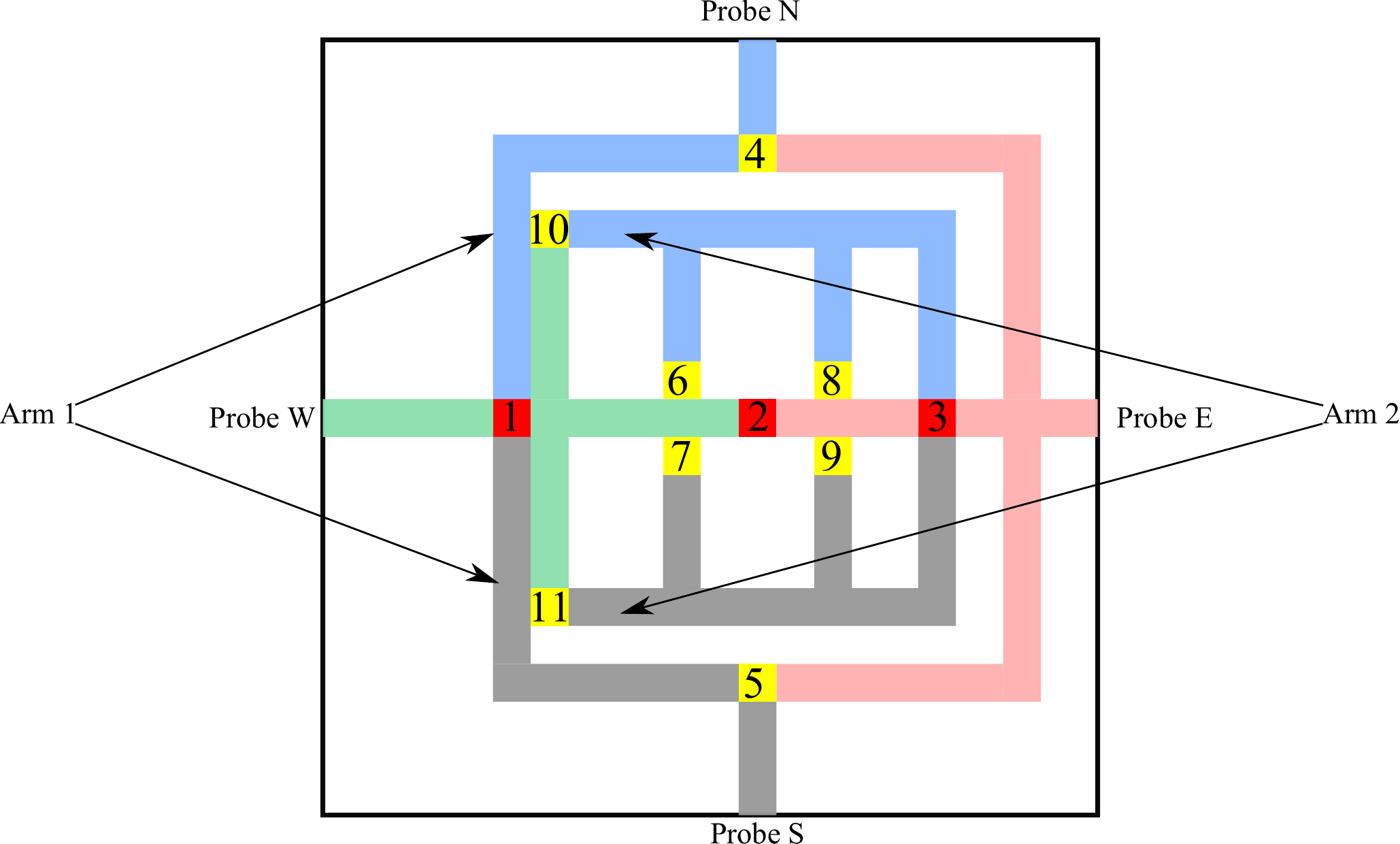}
\caption{Macrotile probes, points of cooperation, and points of competition}
\label{fig:macro_labeled-main}
\end{center}
\vspace{-25pt}
\end{figure}
Informally, the subassemblies labeled probe, which we will now refer to as probes, ``simulate'' the glues of the tiles in $T$.  If a probe is simulating a glue which is of strength $2$, then it does not require the assistance of any other probes in order to complete the macrotile containing it.  On the other hand, if the glue which the probe is simulating is of strength $1$, then the probe cannot assemble a new macrotile until another probe arrives which simulates a glue with which the other glue can cooperate and place a new tile in $\calT$.  Before probes can begin the growth of a new macrotile, they must claim (i.e. place a tile in) one of the \emph{points of competition} (shown as red in Figure~\ref{fig:macro_labeled-main}) depending on the configuration of the macrotile.  Once a special tile is placed in one of the points of competition, the representation function $R$ maps the macrotile to the corresponding tile in $T$, and the growth of the macrotile can begin.

We use the following conventions for our figures.  All duples are shown in darker colors (even after they are broken apart) and singletons are shown in lighter colors.  Negative glues are represented by red squares protruding from tiles, and positive glues are represented by all other colored squares protruding from tiles.  We represent glue mismatches (a glue mismatch occurs when two different glues are adjacent or a glue is adjacent to a tile side that does not have a glue) by showing the mismatching glues receded into the tiles from which they would normally protrude.  A red box enclosing a subassembly indicates that subassembly has total binding strength 0.

\vspace{-15pt}
\begin{figure}[htp]
\begin{center}
\includegraphics[width=3.5in]{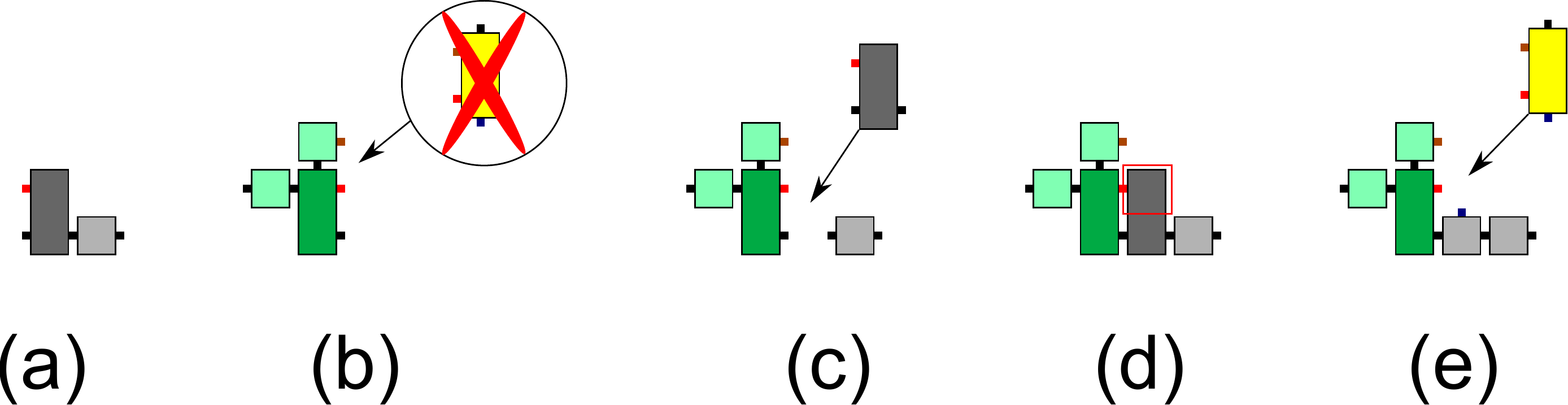}
\caption{An assembly sequence of an adjacent cooperator gadget.}
\label{fig:gad_acoop-main}
\end{center}
\vspace{-25pt}
\end{figure}

The cooperator gadget is the underlying mechanism that allows for the DrgTAM to simulate the cooperative placement of a tile in a $\tau\ge2$ TAS.  We consider two cases of cooperative tile placement: 1) the tiles that cooperatively contribute to the placement of a tile have adjacent corners (e.g. one is north of the location to be cooperatively tiled while the other is to the east or west), and 2) the tiles that cooperatively contribute to the placement of a tile are non-adjacent, that is there is a tile wide gap between the two tiles.  We create a cooperator gadget for each of these two cases.  Not surprisingly, we call the cooperator gadget that mimics the former case the \emph{adjacent cooperator gadget} and the cooperator gadget that mimics the latter case the \emph{gap cooperator gadget}. Each of these two gadgets is asymmetric in nature and consists of two parts: 1) a finger and 2) a resistor.  The function of the resistor is to cause a duple that is attached to the finger gadget to break apart and expose the internal glue of the duple which can then be used for binding of another tile.

An adjacent cooperator gadget is shown in Figure~\ref{fig:gad_acoop-main}.  Part (a) of this figure depicts the finger part of the gadget, and the subassembly labeled (b) is the resistor.  Note that the only tiles which have the ability to bind to the exposed glues are duples with a negative glue that is aligned with the negative glue that is adjacent to the exposed glues.  This means that neither subassembly can grow any further until its counterpart arrives.  In Figure~\ref{fig:gad_acoop-main} parts (c) - (e) we see the assembly sequence showing the interaction between the two parts of the cooperator gadget. In this particular assembly sequence we have assumed that the resistor piece of the gadget has arrived first.  In part (c), we see the arrival of a tile (presumably from a probe) which allows for the duple that is a part of the finger gadget to bind with total strength 1.  The 0 strength cut that is induced by this binding event is shown by the red box in part (d) of the figure.  Since the tile encapsulated in the red box is bound with total strength 0, it eventually detaches which leads us to part (e) of the figure.  Notice that the dissociation event has caused a new glue to be exposed.  This glue now allows for the binding of a duple as shown in part (e) of Figure~\ref{fig:gad_acoop-main}.

\begin{figure}[htp]
\begin{center}\includegraphics[width=4in]{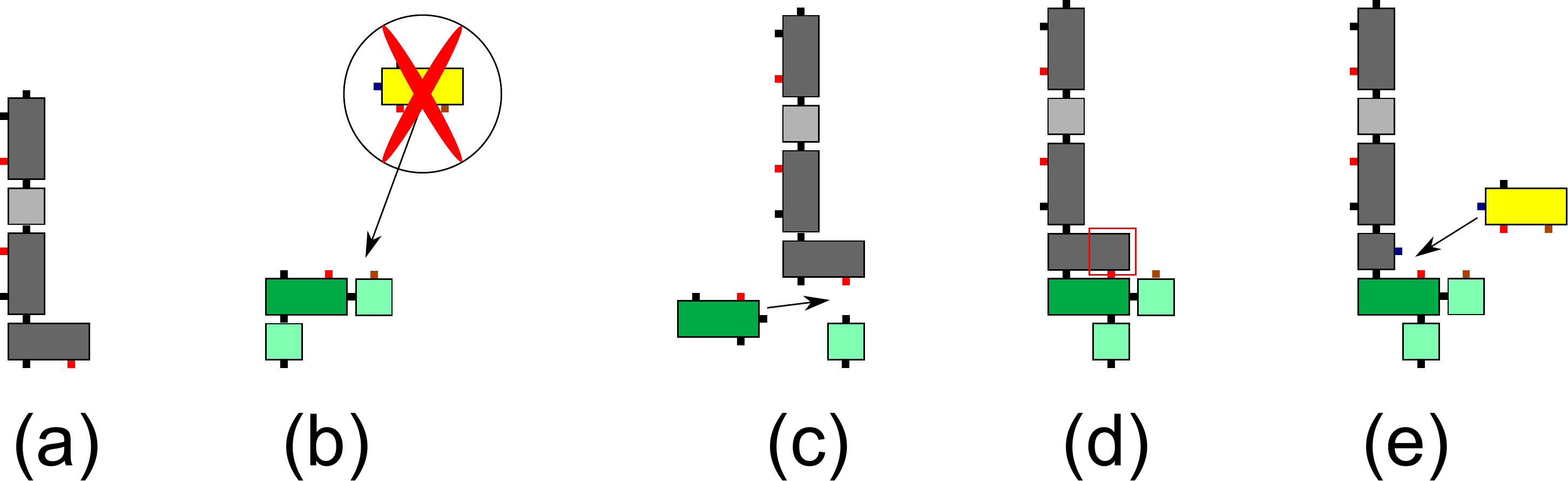}
\caption{An assembly sequence of a gap cooperator gadget.}
\label{fig:gad_coop-main}
\end{center}
\vspace{-25pt}
\end{figure}

Figure~\ref{fig:gad_coop-main} shows a gap cooperator gadget which is a simple extension of the adjacent cooperator gadget.  This extension of the adjacent cooperator gadget allows for a crosser gadget (described below) to grow a path of tiles in between the two parts of the gadget.  This gadget allows a new glue to be exposed upon the arrival of a negative glue (Figure~\ref{fig:gad_coop-main} part (c))  which causes half of the duple to detach (shown in part (d) of the figure).  This allows a duple to attach as shown in Figure~\ref{fig:gad_coop-main}(e) which depends on both of the glues exposed by the two parts of the gadget.  Notice that the binding of this tile cannot occur unless both parts of the gadget are present.

The previous gadgets showed that in order for two probes to cooperate, they must be connected by a path of tiles.  In order for other probes to cross in between these connected probes we utilize what we call a ``crosser gadget''.  The assembly sequence for a crosser is shown in Figure~\ref{fig:gad_crosser-main}.  Growth of the gadget begins with the placement of a singleton which is prevented from growing further.  This singleton exposes glues which allow for duples to bind (Figure~\ref{fig:gad_crosser-main}(b) and (c)) that cause the path of tiles blocking the singleton's growth to detach (Figure~\ref{fig:gad_crosser-main}(d)).  Note that the attachment of these duples cannot occur before the singleton arrives since they would only have total binding strength zero.

Section~\ref{sec:gadgets} offers a more in-depth description of the gadgets described above.

We can now use these gadgets to give a more complete description of the probes which are shown in Figure~\ref{fig:macro_labeled-main}.  All of the numbered regions represent gadgets.  Gadgets labeled 1-3 in the figure represent gap cooperator gadgets which allow for cooperation between the probes to which they are attached.  The gadgets labeled 5-9 denote adjacent cooperator gadgets which allow for the potential of cooperation between the probes to which they are attached.  Finally, the gadgets labeled 10 and 11 are cooperator gadgets which allow for Probe W to trigger the growth of the second arms of Probe N and Probe S.  See Section~\ref{sec:probes} for more details about the structure of probes and their accompanying gadgets.

\vspace{-15pt}
\begin{figure}[htp]
\begin{center}
\includegraphics[width=2.7in]{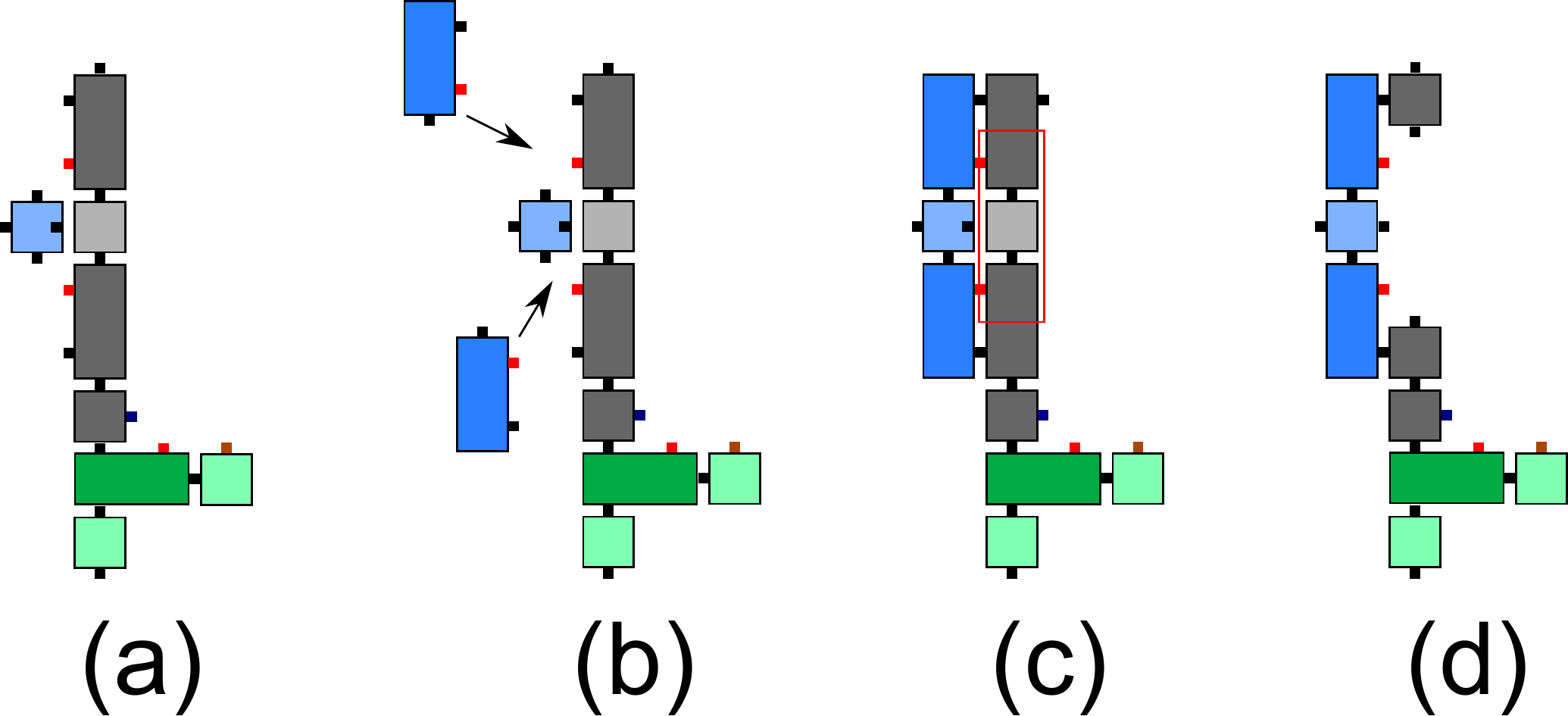}
\caption{An assembly sequence of a crosser gadget.}
\label{fig:gad_crosser-main}
\end{center}
\vspace{-25pt}
\end{figure}

The output of the representation function for a particular macrotile depends on the three regions labeled 1-3 in Figure~\ref{fig:macro_labeled-main}.  If a special tile is placed in region 1, then the macrotile region is mapped to the tile in $T$ that corresponds to the special tile regardless of the tiles in the other regions.  Similarly, region 3 takes precedence over region 2.  Finally, if a special tile has not been placed in either region 1 or 3, then the output of the representation function depends on the tile placed in region 2.  For a more detailed explanation of the representation function and regions 1-3 see Section~\ref{sec:poc_repr}.  For a case analysis of how our construction handles all possible binding scenarios, see Section~\ref{sec:case_analysis}.

The seed of our simulator is formed from a set of tiles in $S \cup D$ which have been hardcoded.  Section~\ref{sec:seed_form} gives a more detailed explanation about the construction of the seed in the simulator.

\ifabstract
\later{
\section{Gadgets: Cooperators and Crossers} \label{sec:gadgets}
We now introduce two gadgets which give the probes mentioned above the required functionality needed to imitate cooperatively placing a tile.  Furthermore, these gadgets will allow us to modularize the construction in the proceeding sections.  The first gadget that we introduce is called the cooperator gadget.  As its name suggests, its purpose is to mimic the cooperation found in $\tau = 2$ TASs. The second gadget we describe, which we call the crosser gadget, allows for probes to cross in between each other.  For example, a crosser gadget enables the east and west probes to grow through the north and the south probes.

A key observation to make during the description of these gadgets is that these gadgets are designed such that all tiles that detach from the assembly are singletons that are originally part of a duple unless otherwise specified.  We construct the duples such that this glue is unique, and consequently nothing can bind to the portion of the duple that fell off of the assembly except its counterpart which is attached to the assembly.  Indeed, observe that any duple presents at most one negative glue.  This implies that the same half is always the one which detaches, and consequently there are not any tiles which may bind to it. This means that all of the tiles that detach from the assembly are inert (i.e. unable to bind to any tile in solution).  If tiles that fell off the assembly were not inert, then it could be possible to grow assemblies which would invalidate the simulation (since the definition of simulation requires that so-called junk assemblies must never grow into assemblies which map to something other than the empty tile under $R$).  Thus, it is necessary that we be careful about what we allow to detach from the assembly.

Throughout this section, we use the following conventions for our figures.  All duples are shown in darker colors (even after they are broken apart) and singletons are shown in lighter colors.  Negative glues are represented by red squares protruding from tiles, and positive glues are represented by all other colored squares protruding from tiles.  We represent glue mismatches (a glue mismatch occurs when two different glues are adjacent or a glue is adjacent to a tile side that does not have a glue) by showing the mismatching glues receded into the tiles from which they would normally protrude.

\subsection{Cooperators}
The cooperator gadget is the underlying mechanism that allows for the DrgTAM to simulate the cooperative placement of a tile in a $\tau\ge2$ TAS.  As in the aTAM at $\tau=2$, cooperator gadgets allow the attachment of tiles in one subassembly to trigger growth in another subassembly.  We consider two cases of cooperative tile placement: 1) the tiles that cooperatively contribute to the placement of a tile have adjacent corners (e.g. one is north of the location to be cooperatively tiled while the other is to the east or west), and 2) the tiles that cooperatively contribute to the placement of a tile are non-adjacent, that is there is a tile wide gap between the two tiles.  We create a cooperator gadget for each of these two cases.  Not surprisingly, we call the cooperator gadget that mimics the former case the \emph{adjacent cooperator gadget} and the cooperator gadget that mimics the latter case the \emph{gap cooperator gadget}. Each of these two gadgets are asymmetric in nature and consist of two parts: 1) a finger and 2) a resistor.  The function of the resistor is to cause a duple that is attached to the finger gadget to break apart and expose the internal glue of the duple which can then be used for binding of another tile.

An adjacent cooperator gadget is shown in Figure~\ref{fig:gad_acoop}.  Part (a) of this figure depicts the finger part of the gadget, and the subassembly labeled (b) is the resistor.  Note that the only tiles which have the ability to bind to the exposed glues are duples with a negative glue that is aligned with the negative glue that is adjacent to the exposed glues.  This means that neither subassembly can grow any further until its counterpart arrives.  In Figure~\ref{fig:gad_acoop} parts (c) - (e) we see the assembly sequence showing the interaction between the two parts of the cooperator gadget. In this particular assembly sequence we have assumed that the resistor piece of the gadget has arrived first.  In part (c), we see the arrival of a tile (presumably from a probe) which allows for the duple that is a part of the finger gadget to bind with total strength 1.  The 0 strength cut that is induced by this binding event is shown by the red box in part (d) of the figure.  Since the tile encapsulated in the red box is bound with total strength 0, it eventually detaches which leads us to part (e) of the figure.  Notice that the dissociation event has caused a new glue to be exposed.  This glue now allows for the binding of a duple as shown in part (e) of Figure~\ref{fig:gad_acoop}.

\begin{figure}[htp]
\begin{center}
\includegraphics[width=3.5in]{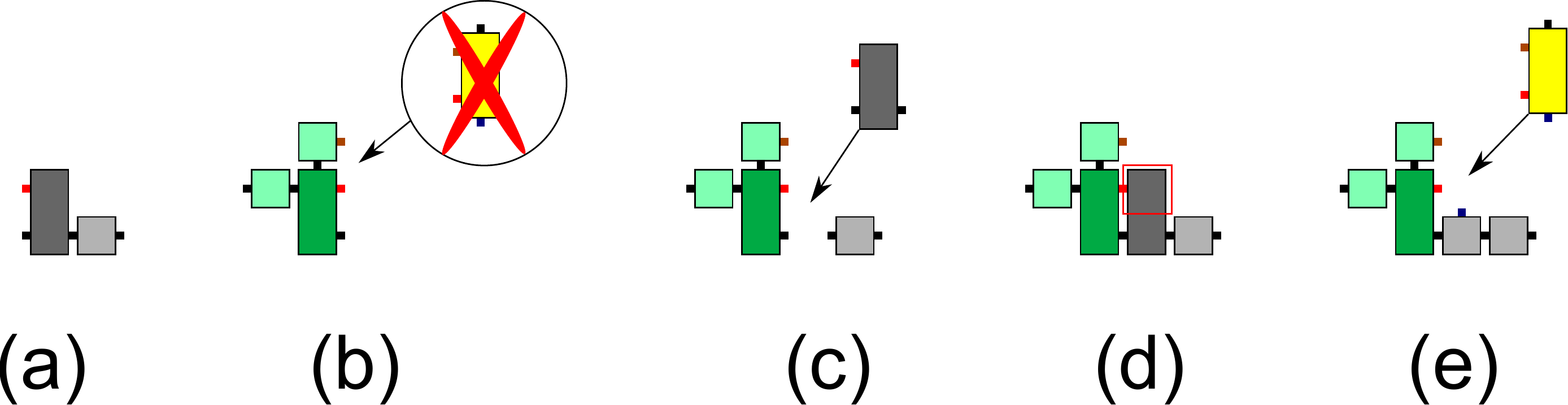}
\caption{An assembly sequence of an adjacent cooperator gadget.}
\label{fig:gad_acoop}
\end{center}
\end{figure}

We now present an example which demonstrates an adjacent cooperator gadget.  Suppose that $T$ contains the subassembly shown in Figure~\ref{fig:gad_coop_exT} (a), and the only tiles which may bind to the west glue of tile $A$ are shown in part (b) of the figure.  Observe that since we are in a system of temperature 2, only tile $C$ may bind to this subassembly.  Tile $D$ cannot bind because its binding strength to this subassembly is 1.  Part (c) of Figure~\ref{fig:gad_coop_exT} shows the subassembly after tile $C$ binds which is the only binding event that can occur at that location.  Figure~\ref{fig:gad_coop_exS} shows the assembly sequence of the adjacent cooperator gadget which simulates the binding event that occurs in Figure~\ref{fig:gad_coop_exT}.  Note that the parts of the cooperator gadget lie in the macrotile region that eventually contains a macrotile which maps to tile $C$ under the representation function.  Part (a) of this figure shows the two tiles which allow for the growth of a macrotile to begin which maps to either tile $C$ or $D$ in $T$.  Parts (b) and (c) show the assembly sequence which leads us to the subassembly shown in part (d).  The subassembly in part (d) makes it such that the tile which is placed where the arrows are pointing must have both glues match the two glues exposed by the gadget.  This ensures simulation of the binding of the tile labeled $C$.
\begin{figure}[htp]
\begin{center}
\includegraphics[width=2.5in]{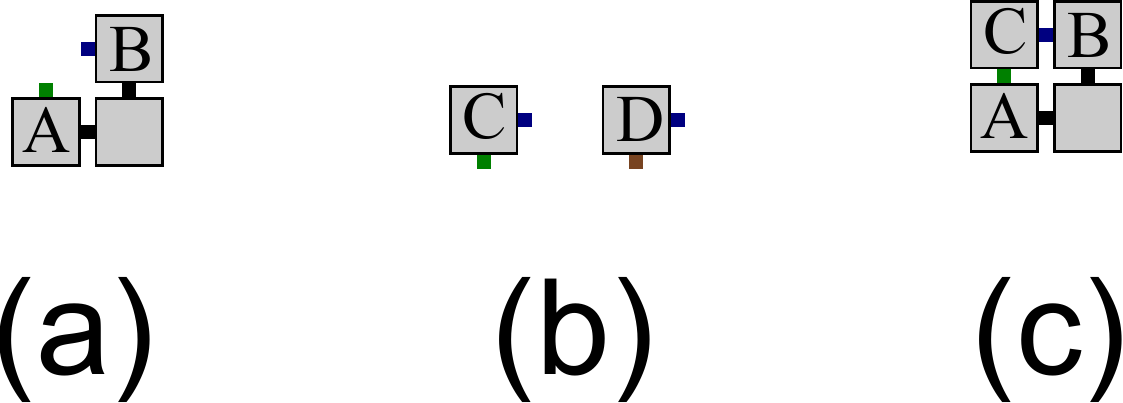}
\caption{An example subassembly sequence in $\calT$.}
\label{fig:gad_coop_exT}
\end{center}
\end{figure}

\begin{figure}[htp]
\begin{center}
\includegraphics[width=2.5in]{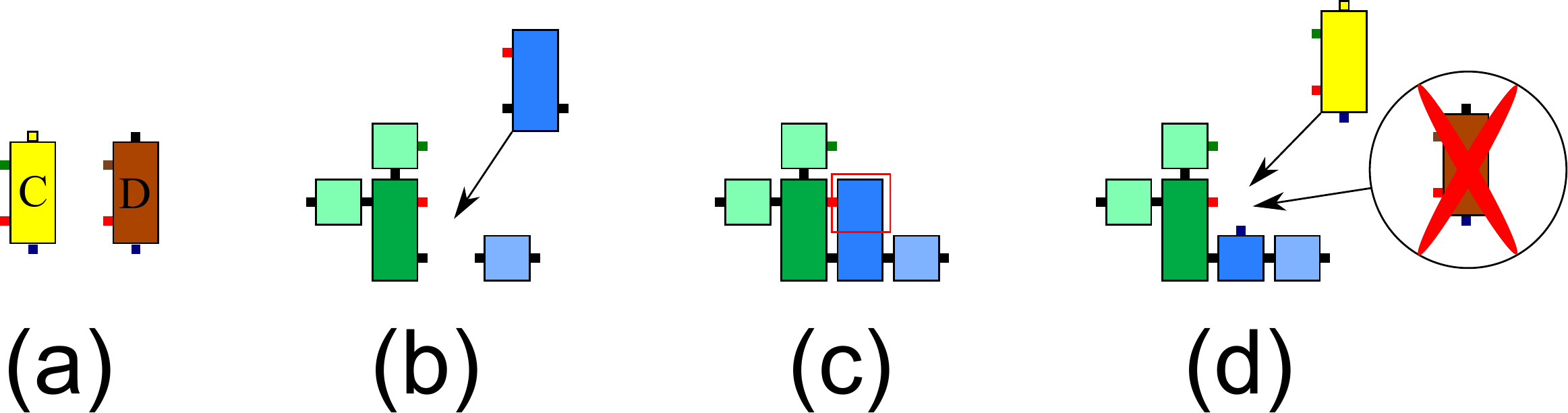}
\caption{Using an adjacent cooperator gadget to mimic the cooperative tile placement shown in Figure~\ref{fig:gad_coop_exT}.}
\label{fig:gad_coop_exS}
\end{center}
\end{figure}

Figure~\ref{fig:gad_coop} shows the finger part of the gap cooperator in part (a) and the resistor portion in part (b). Notice that the end of the finger gap cooperator gadget has the same structure as the finger portion of the adjacent cooperator, and the two resistor parts of the gadgets are equivalent as well.  The only difference between the two gadgets is that the finger gap cooperator gadget consists of an extra three tiles which precede the duple exposed to the resistor part of the gadget.  In the next section, we will see that these extra tiles are necessary in order for the crosser gadget to be implemented.  Parts (c)-(e) of Figure~\ref{fig:gad_coop}, show the assembly assembly sequence of a gap cooperator when its two pieces interact.

\begin{figure}[htp]
\begin{center}
\includegraphics[width=4.5in]{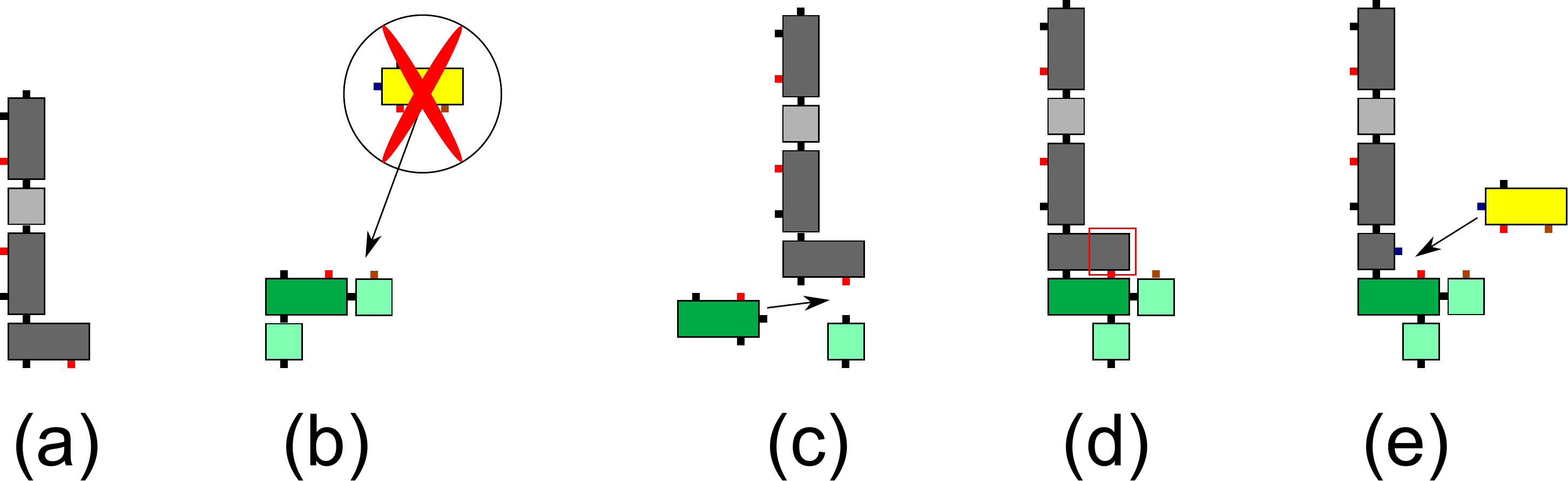}
\caption{An assembly sequence of a gap cooperator gadget.}
\label{fig:gad_coop}
\end{center}
\end{figure}

\subsection{Crossers}
As we saw in the previous section, the only ways for probes to mimic cooperation requires them to be connected by a tile wide path.   In order for other probes to cross in between these connected probes we utilize what we call a crosser gadget.  The assembly sequence for a crosser is shown in Figure~\ref{fig:gad_crosser}.  Growth of the gadget begins in part (a) of the figure which shows a singleton arriving at a gap cooperator which is described above.  Upon the arrival of this singleton, two duples may be placed with total binding strength one as shown in part (b).  Note that the attachment of these duples cannot occur before the singleton arrives since they would only have total binding strength zero.  The attachment of these two duples induces a strength zero cut which contains the subassembly inside of the red box shown in part (c) of Figure~\ref{fig:gad_crosser}.  Since this cut of the binding graph has total strength zero, the subassembly inside of the red box will eventually detach which leads to the assembly shown in part (d) of the figure.  Now, it is possible for the single tile wide path to continue its growth to the other side of the probes.

Unlike the other gadget we explored, this gadget allows for subassemblies which consist of more than just one half of a duple to detach from an assembly.  Figure~\ref{fig:gad_junk} shows all of the subassemblies which can detach due to the crosser gadget.  Notice that the ``junk'' in part (a) of this figure is inert since the two exposed glues are unique internal duple glues.  This is the only thing that can detach in the situation that we explored above where the finger gap cooperator gadget is bound to the resistor gadget which total strength 1.  But, it could be the case that the finger gadget is in the process of growing or that the probe to which the resistor gadget is attached never arrives.  This situation gives rise to the junk shown in parts (b)-(e).  Observe that the ``junk'' in parts (b)-(d) can only grow into (e) which is also inert.  Thus, anything that that the crosser gadget causes to detach is inert. Note that (b) and (c) do not grow into (a) since the bottom tile of the subassembly in (a) is half of a duple which was broken apart by the crosser gadget.  This duple half cannot attach to any assembly without its counterpart, thus only a full duple can attach to the subassemblies (b) and (c) which causes them to grow into (e).

\begin{figure}[htp]
\begin{center}
\includegraphics[width=3.0in]{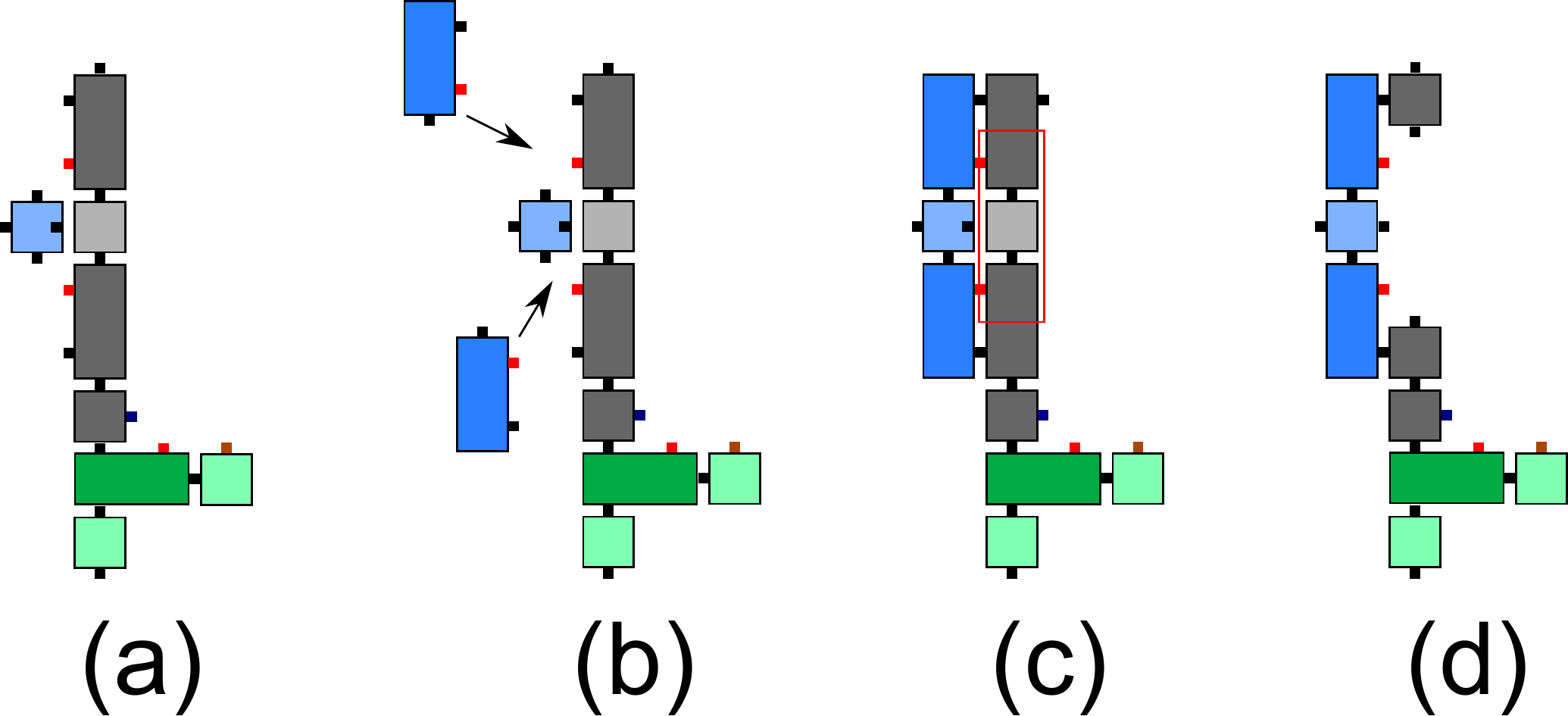}
\caption{An assembly sequence of a crosser gadget.}
\label{fig:gad_crosser}
\end{center}
\end{figure}

\begin{figure}[htp]
\begin{center}
\includegraphics[width=2.0in]{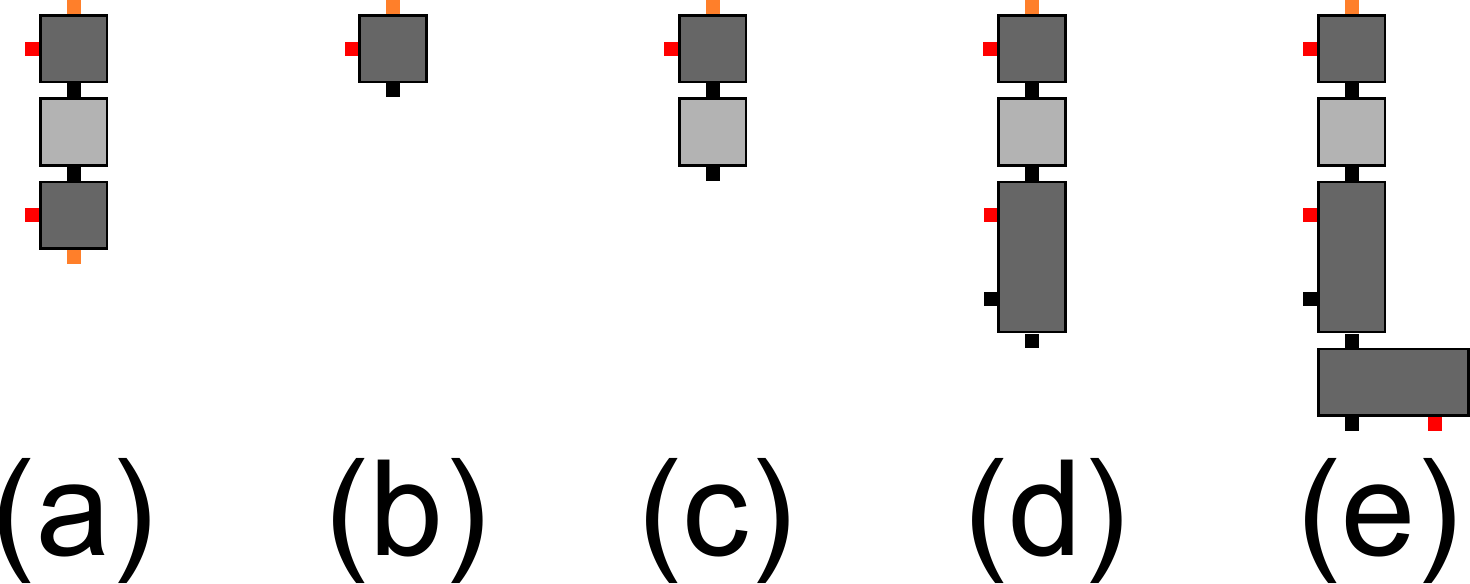}
\caption{The subassemblies that can detach due to the crosser gadget. }
\label{fig:gad_junk}
\end{center}
\end{figure}

\subsection{Interfacing Gadgets}
Now, with these gadgets in our toolbox, we can discuss how they will be utilized in our construction.  First, notice that we can orient these gadgets however we see fit by rotating or flipping the tiles from which they are assembled.  In addition, for our construction we will not use the gap cooperator gadget which we described above, but rather the extension of it shown in Figure~\ref{fig:gad_coop_complex}(a).  This is necessary since, as we will see, we need paths of tiles to cross through the gap cooperators from either direction.  Consequently, we must add another set of three tiles which will allow for crossers coming from either direction to cross through the gadget.  Part (b) of this figure shows a cooperatively placed tile growing a crosser gadget in order to begin growing its probes.  Observe that whenever this occurs, a probe trying to grow southward will be prevented from growing due to the path of tiles which were laid down by the cooperatively placed tile.  But, this is not an issue since at this point the tile in $T$ which the macrotile simulates has already been decided.  Thus, there is no need for any other probes in the macrotile region to grow or cooperate with each other.
\begin{figure}[htp]
\begin{center}
\includegraphics[width=2.0in]{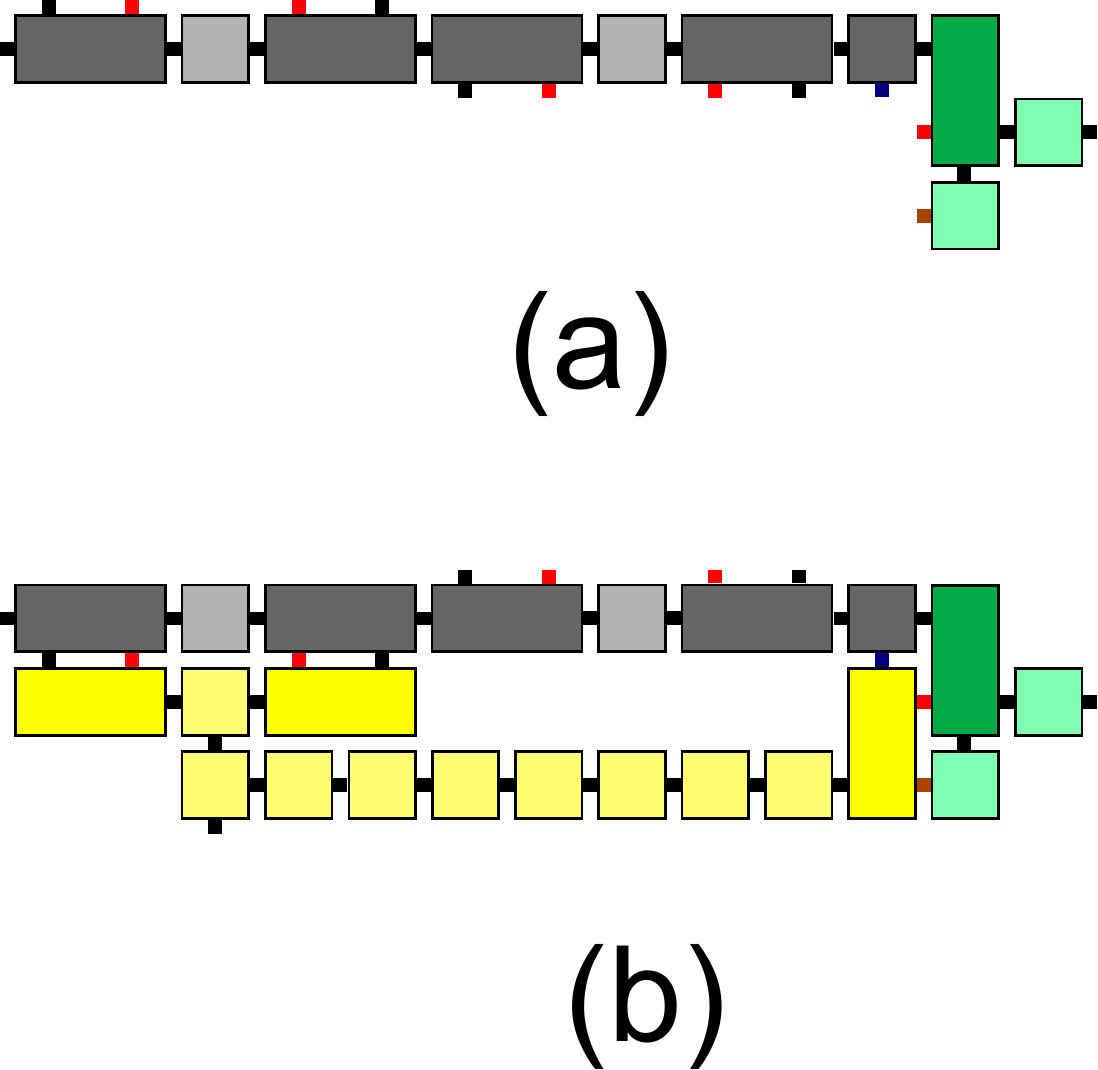}
\caption{The extension of the gap cooperator which our construction will use.}
\label{fig:gad_coop_complex}
\end{center}
\end{figure}

\section{Probe Configurations}\label{sec:probes}
Probes can take on multiple configurations depending on the strength of the glue they are simulating and the probes already in the macrotile region when they arrive.  All probes consist of a single-tile wide path of tiles to which the gadgets describe above are attached.  There are two types of fundamentally different probes: probes that grow from the north and south of the macrotile and probes that grow from the east and west of the macrotile. As shown in Figure~\ref{fig:macro_labeled}, probes that grow from the east and west are single arm probes while probes that grow from the north and south potentially require two arms.  Figure~\ref{fig:macro_labeled} also shows all of the probes along with their corresponding gadgets which are marked as colored number regions.  Gadgets labeled 1-3 in the figure represent gap cooperator gadgets which allow for cooperation between the probes to which they are attached.  The gadgets labeled 5-9 denote adjacent cooperator gadgets which allow for the potential of cooperation between the probes to which they are attached.  Finally, the gadgets labeled 10 and 11 are cooperator gadgets which allow for Probe W to trigger the growth of the second arms of Probe N and Probe S.

\section{Points of Competition and the Representation Function}\label{sec:poc_repr}
Before probes can place tiles which begin the growth of a particular macrotile, the probes must grow paths which claim a point of competition.  Once a point of competition is claimed by a special tile, the representation function can then map the macrotile to a tile in $T$.

\begin{figure}[htp]
\begin{center}
\includegraphics[width=2.5in]{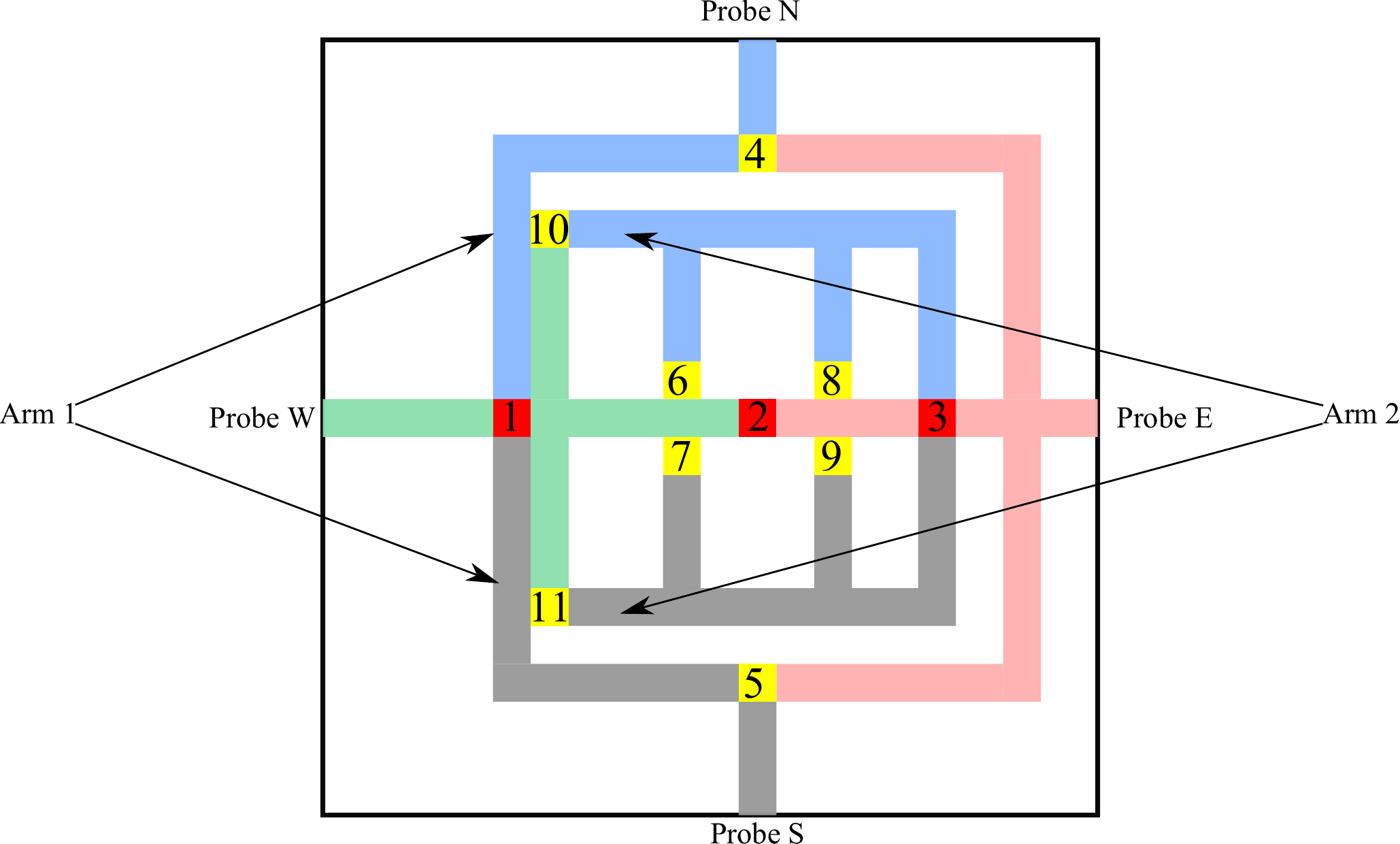}
\caption{A schematic picture of the probes, points of cooperation, and points of competition of a macrotile.}
\label{fig:macro_labeled}
\end{center}
\end{figure}

Figure~\ref{fig:macro_labeled} gives a schematic picture of the paths that probes take as they assemble as well as the location of the points of cooperation and points of competition. To simulate growth of $\tau=2$ aTAM systems, we must handle two cases: (1) a tile binds via a strength-$2$ glue, (2) a tile binds via the cooperation of two strength-$1$ glues.

When simulating case (1), as a macrotile assembles, a probe representing a strength-$2$ glue claims a point of competition (labeled $2$ in Figure~\ref{fig:macro_labeled}) by placing a special tile in  a designated location before any other probe can place a tile at the same location. Once this special tile is placed, subassemblies form by single tile additions starting from a glue exposed by the special tile. These subassemblies output glues on the relevant sides of the macrotile. We call such subassemblies \emph{glue outputting} subassemblies. A glue outputting subassembly may attempt to present glues on a side of the macrotile where a probe has started to assemble. In this case the glue outputting subassembly simply crashes into the (possible partially formed) probe. In Figure~\ref{fig:macro_labeled}, none of the glue outputting subassemblies are shown, only the various probes are shown. See Section~\ref{sec:case_analysis} for detailed analysis of each case of simulating strength-$2$ binding.

More interesting cases arise when simulating case (2). We will give a high-level description of each case of cooperation here. See Section~\ref{sec:case_analysis} for complete details. First, each probe uses unique glues to assemble for each glue in $\calT$. Denote the glue that Probe D represents by $g_D$, where $D$ is one of $N$, $S$, $E$, or $W$. We will see that special duples attached to probes can be placed to win points of competition (specially designated tile locations of a macrotile). In winning these locations, these duples determine which tile is being simulated.

To simulate the cooperation of glues $g_N$ and $g_S$, Probe N and Probe S can win the point of cooperation at the region with label $1$ in Figure~\ref{fig:macro_labeled}. If these two probes indeed cooperate, then appropriate glue outputting subassemblies form. Notice that Probe W may occupy tile locations in region $1$ before Probe N and Probe S have a chance to cooperate. So that this does not prevent the simulation of cooperation of glues $g_N$ and $g_S$, when Probe W crosses region $1$ (using a crossing gadget), it uses adjacent cooperator gadgets to allow secondary probes to form from Probe N and Probe S. Note that these particular adjacent gadgets which trigger the growth of the second arm of the probe are generic.  That is, all west probes, regardless of which glue they are simulating, present the same cooperator gadgets to trigger the growth of the second arm of the south and north probes.  These secondary probes can then cooperate at region $3$. If they do, a glue outputting subassembly forms to present glues on the east side of the macrotile. Thinking of regions $1$ and $3$ as points of competition, when Probe N and Probe S win a special tile location in either region, the representation function maps the macrotile to a tile type in $T$ based on the special tile placed in either region $1$ or $3$.

Similarly, to simulate cooperation of glues $g_E$ and $g_W$, Probe E and Probe W can cooperate at the region labeled $2$ in Figure~\ref{fig:macro_labeled}. At this point, glue outputting subassemblies attempt to present glues to the north and south sides of the macrotile.

Simulation of cooperation of glues $g_W$ and $g_S$ is equivalent up to reflection to simulation of cooperation of glues $g_W$ and $g_N$. We will describe the cooperation of Probe W and Probe S. For Probe W and Probe S to cooperate, Probe W must first cross region $1$ and trigger the growth of secondary probes for Probe S. Using one of these secondary probes, Probe W and Probe S may cooperate at region $7$. Once cooperation has occurred in this region, a path of tiles assembles toward region $2$. If this path of tiles places a tile in a specially designated tile location (a point of competition) of region $2$, appropriate glue outputting subassemblies may form.

Finally, simulation of cooperation of glues $g_E$ and $g_S$ is equivalent up to reflection of simulation of cooperation of glues $g_E$ and $g_N$. Therefore, we only describe cooperation of Probe E and Probe S. Probe E and Probe S may cooperate at region $5$ or region $9$. If cooperation occurs at region $5$, a path of tiles binds one tile at a time until the point of competition in region $2$ is won, at which point, appropriate glue outputting subassemblies may form. Notice that Probe W may have triggered the growth of secondary probes from Probe S. If this is the case, these secondary probes may prevent the formation of the path of tiles that would otherwise be able to claim the point of competition in region $2$. For this reason, Probe S and Probe E may also cooperate in region $9$, at which point a path of tiles forms, claims a point of competition in region $2$, and glue outputting subassemblies form.

\section{Case Analysis of Tile Placements in $\calT$}\label{sec:case_analysis}
We now look at how our simulator is able to simulate every possible way a tile could attach in $\calT$.  In order to accomplish this, we need to only look at $18$ informative cases.  The other cases will follow from the symmetry of our construction. When tiles bind in $\calT$, they may do so by either attaching with a strength-$2$ glue or they may do so by the cooperation of two strength-$1$ glues. When simulating $\calT$, macrotiles that form must take into account the fact that some input glues are not used due to either mismatching or overbinding (i.e. binding with strength greater than $\tau$). Such input glues are called \emph{non-contributing input supersides}. These are input supersides that are not used to simulate tile binding. Mismatching supersides are one such example. We will describe how tile binding in $\calT$ is simulated using macrotiles and make special mention to the cases where there are non-contributing input supersides. Finally, for the remainder of this section, we denote the glue that Probe D represents by $g_D$, where $D$ is one of $N$, $S$, $E$, or $W$, and in the figures for the various cases, we denote the points of competition and points of cooperation in a region labeled $k$ by POC$k$, where $k\in \mathbb{N}$; whether or not POC$k$ is a point of competition or a point of cooperation will be clear from the context.

\subsection{One-sided binding}\label{sec:onesided_binding}

One-sided binding occurs in $\calT$ when a tile binds using a strength-$2$ glue. For example, Figure~\ref{fig:onesided} depicts the attachment of a tile due to the binding of a strength-$2$ east glue of the attaching tile.

\begin{figure}[htp]
\begin{center}
\includegraphics[width=1.5in]{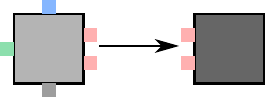}
\caption{Binding via a strength-$2$ glue with no non-contributing input supersides. The tile on the left binds to an assembly using a strength-$2$  east glue.}
\label{fig:onesided}
\end{center}
\end{figure}

To simulate this type of binding, when a strength-$2$ probe grows into an otherwise empty macrotile region (See Figure~\ref{fig:onesided_supertile} for an example of a probe grown from the east.), it grows a path tiles toward the point of competition labeled $2$ in Figure~\ref{fig:onesided_supertile}. If this probe wins this point of competition it places a tile that determines which glues to output on the south, west, and north sides of the macrotile and grows these output glues toward their respective sides. Figure~\ref{fig:onesided_supertile} shows this growth.

\begin{figure}[htp]
\begin{center}
\includegraphics[width=1.5in]{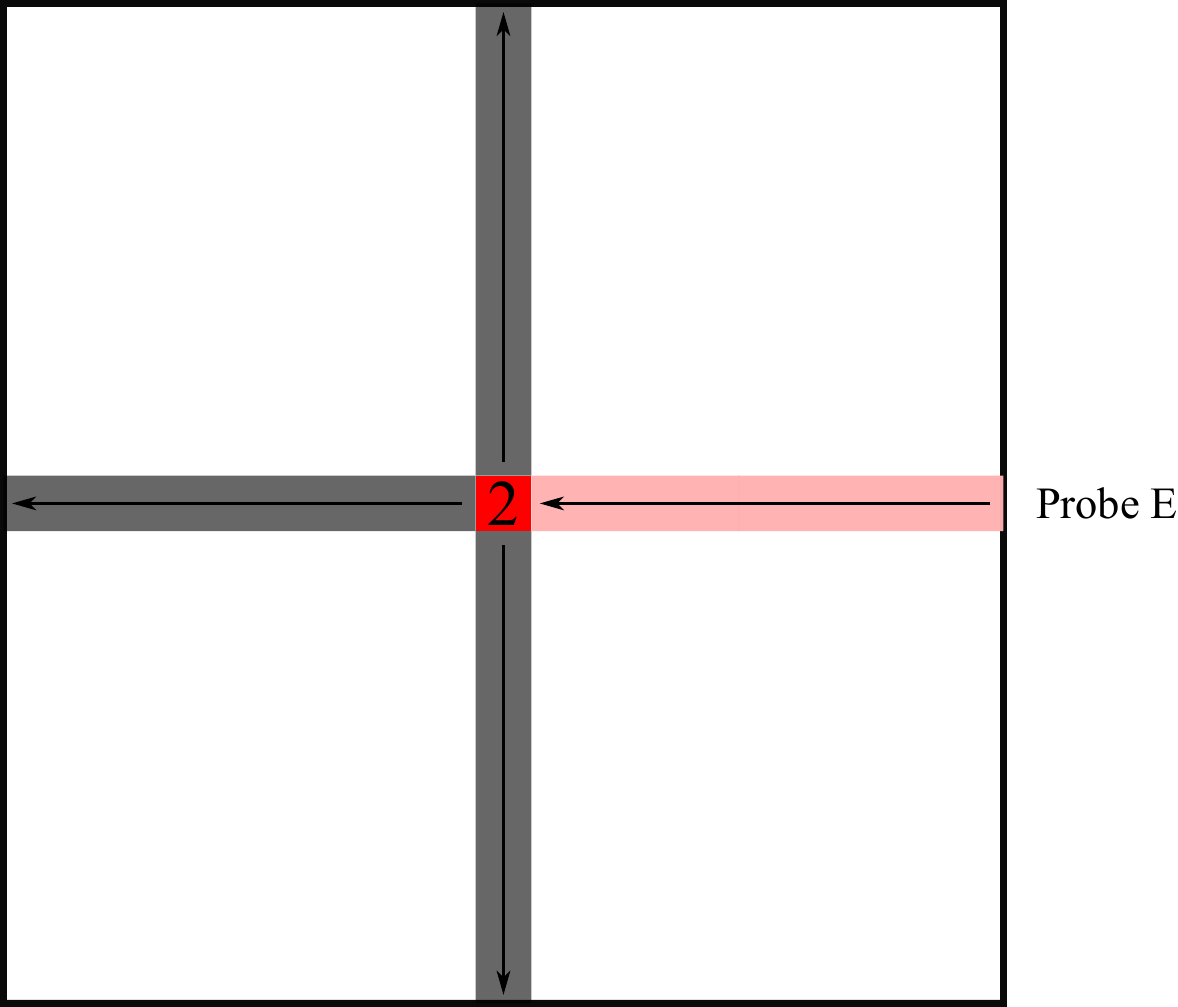}
\caption{Growth of a macrotile that simulates the binding in~\ref{fig:onesided}.}
\label{fig:onesided_supertile}
\end{center}
\end{figure}

For strength-$2$ glues, there are $3$ other cases to consider that are all equivalent to the case in Figure~\ref{fig:onesided_supertile} up to rotation.

\subsection*{One-sided binding with non-contributing input sides}

Now we consider cases of one-sided binding with one or two non-contributing input sides. We consider the three cases of tile binding in $\calT$ depicted in Figure~\ref{fig:onesided_noncontributing} as the rest of the cases are similar to these cases. In each case of Figure~\ref{fig:onesided_noncontributing}, a strength-$2$ glue allows for a tile to bind while either a mismatch or overbinding occurs with the other glues.

\begin{figure}[htp]
\begin{center}
\includegraphics[width=2.5in]{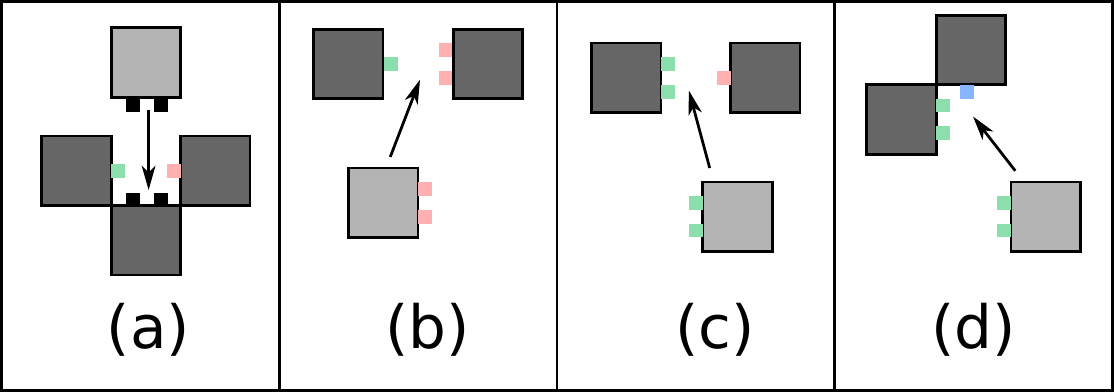}
\caption{Binding via a strength-$2$ glue with non-contributing input supersides.}
\label{fig:onesided_noncontributing}
\end{center}
\end{figure}

In the simulation of $\calT$, special care must be taken to ensure that the probes corresponding to glue mismatching or glue overbinding do not interfere with the growth of a macrotile that is simulating a strength-$2$ tile attachment. For example, when simulating the type of tile attachment shown in Figure~\ref{fig:onesided_noncontributing} part (a), probes enter the macrotile region from the east and west. If the probe from the south wins the point of competition labeled $2$ in Figure~\ref{fig:macro_labeled}, then the east and west probes should not prevent the output of a glue to the north side of the macrotile.

\begin{figure}[htp]
\begin{center}
\includegraphics[width=5.0in]{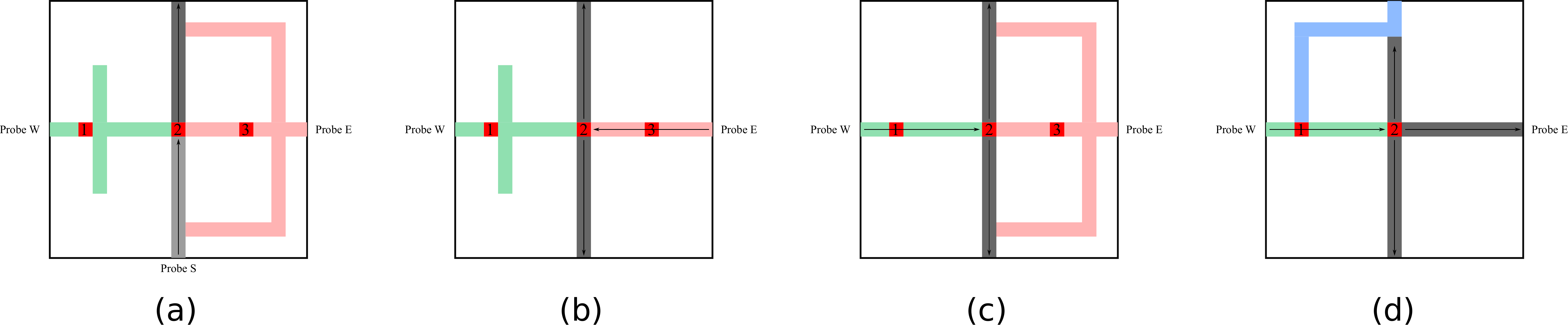}
\caption{Growth of a macrotile that simulates the binding in~\ref{fig:onesided_noncontributing}.}
\label{fig:onesided_noncontributing_supertile}
\end{center}
\end{figure}

In case (a) of Figure~\ref{fig:onesided_noncontributing_supertile}, when Probe S wins POC2, a tile is placed that determines the output glues to be grown to the east, north, and west sides of the macrotile. Since probes have begun growth from the east and/or west (labeled Probe E and Probe W), growth of subassemblies that present glues to the east and west sides of the macrotile will be interrupted by the growth of Probe E and Probe W. If Probe E and Probe W fully form but do not cooperate they do so with a gap cooperator gadget; therefore, using a crossing gadget, Probe S can still cross these probes. A glue outputting subassembly that presents a glue to the north is allowed to assemble since the assembly sequence of this subassembly can be hardcoded to avoid Probe E and Probe W subassemblies. Similarly, in cases (b), (c), and (d) of Figure~\ref{fig:onesided_noncontributing_supertile}, once the point of competition labeled $2$ is won by a strength-$2$ probe (Probe E in case (b) and Probe W in cases (c) and (d)), subassemblies form that output glues to the appropriate sides of the macrotile. Any probes forming from non-contributing input sides prevent the output of glues on those sides.

\subsection{Two-sided binding}

Now we consider the cases where a tile of $\calT$ binds to two tiles via the cooperation of two strength-$1$ glues. We will first consider the cases in Figure~\ref{fig:twosided}. The four cases in Figure~\ref{fig:twosided} cover all cooperative binding cases where there is no non-contributing side that forms. Technically there are two more cases where a north glue cooperates with a west glue (or east glue) to place a tile; however, these cases are equivalent to cases (c) and (d) in Figure~\ref{fig:twosided} since in these cases, the formation of a macrotile that represents a tile in $\calT$ is symmetric about the horizontal line through the center of the macrotile.

\begin{figure}[htp]
\begin{center}
\includegraphics[width=2.5in]{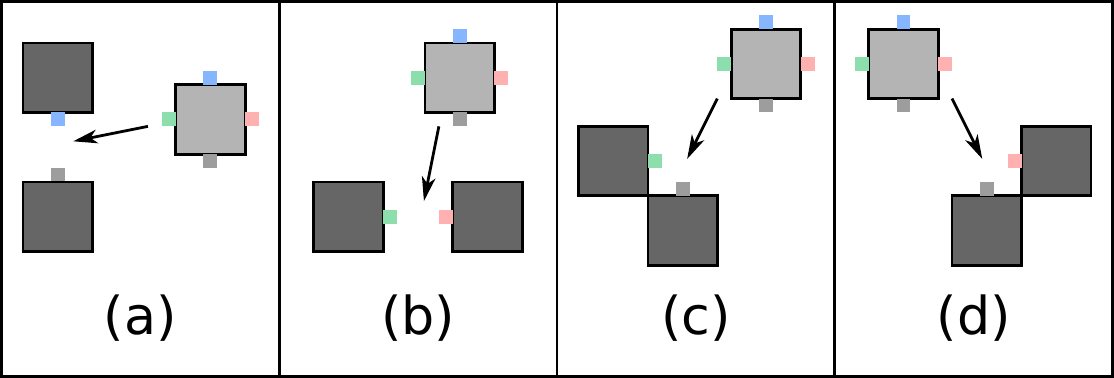}
\caption{Binding via cooperation of two strength-$1$ glues.}
\label{fig:twosided}
\end{center}
\end{figure}

Figure~\ref{fig:twosided_supertile} gives a schematic image for the simulation of the four cases of Figure~\ref{fig:twosided}. In each case, as the macrotile forms, the strategy is essentially the same. When two probes meet at a \emph{point of cooperation}, a cooperator gadget is used to mimic the cooperation that occurs in aTAM systems. There are two types of cooperator gadgets; Section~\ref{sec:gadgets} gives a detailed description of how each cooperator gadget works.

\begin{figure}[htp]
\begin{center}
\includegraphics[width=3.5in]{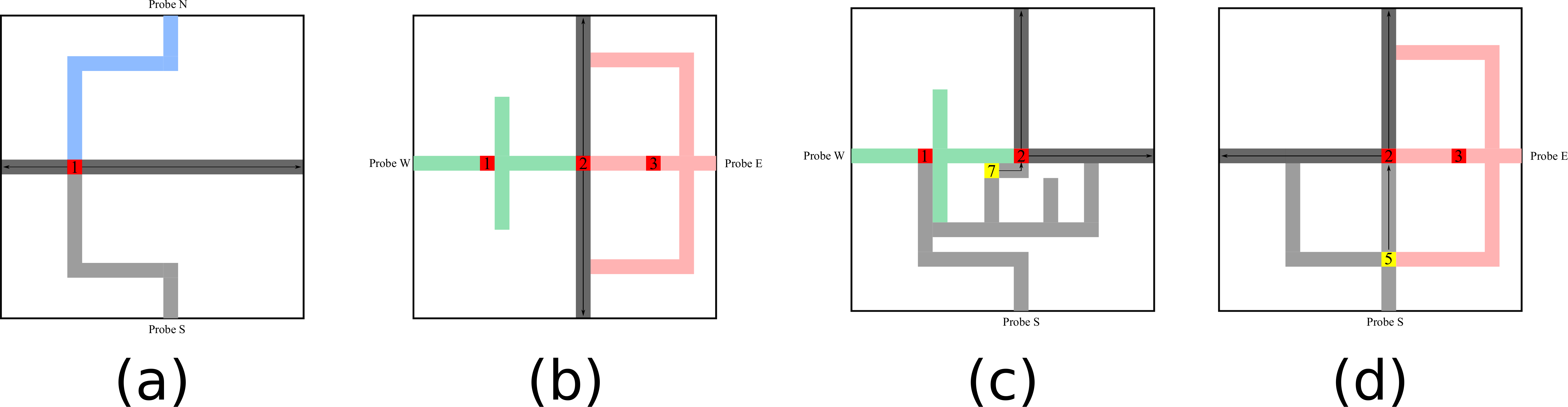}
\caption{Growth of a macrotile that simulates the binding in~\ref{fig:twosided}.}
\label{fig:twosided_supertile}
\end{center}
\end{figure}

\noindent{\bf Case (a):}
In this case, Probe N and Probe S meet at POC1 in Figure~\ref{fig:twosided_supertile}(a). A \emph{gap cooperator} gadget is used to allow for the placement of a duple if and only if there exists a tile type in $T$ with north glue $g_N$ and south glue $g_S$. There is a unique duple for each such tile type in $T$ and the binding of one of these duples allows for the growth of subassemblies that output glues corresponding to the glues on the east and west of the relative tile type in $T$.
\newline

\noindent{\bf Case (b):}
Probe E and Probe W simulate cooperative binding when they meet at POC2 in Figure~\ref{fig:twosided_supertile}(b). A gap cooperator gadget is used to allow for the placement of a duple if and only if there exists a tile type in $T$ with east glue $g_E$ and west glue $g_W$. As in case (a), this duple determines which glue outputting subassemblies form to present glues on the north and south sides of the macrotile.
\newline

\noindent{\bf Case (c):}
Probe S and Probe W simulate cooperative binding as follows. First, Probe W wins POC1. After it wins this point, it grows a subassembly to the north and south of POC1 and to the east of where Probe S assembles. This subassembly uses an adjacent cooperator gadget to trigger Probe S to assemble secondary probes. One of these probes can cooperate with Probe W at POC7. An \emph{adjacent cooperator} gadget is used to allow for the placement of a duple if and only if there exists a tile type in $T$ with west glue $g_W$ and south glue $g_S$. Again, there is a unique duple for each such tile type in $T$ and the binding of one of these duples allows for the growth of subassemblies that output glues corresponding to the glues on the east and north of the relative tile type in $T$. It is at POC7 that a duple is placed that determines which east and north glues to present. Once this duple is placed, growth continues toward POC2. Upon winning POC2, glue outputing subassemblies form that present glues on the east and north sides of the macrotile.
\newline

\noindent{\bf Case (d):}
Probe E and Probe S simulate cooperative binding when they meet at POC5 in Figure~\ref{fig:twosided_supertile}(d). An adjacent cooperator gadget is used to allow for the placement of a duple if and only if there exists a tile type in $T$ with east glue $g_E$ and south glue $g_S$. This duple determines which glue outputting subassemblies form to present glues on the north and west sides of the macrotile. Once this duple is placed tiles attach that race toward POC2. Upon winning this point of competition, the glue outputting subassemblies form.

\subsection*{Two-sided binding with non-contributing input sides}

Here we present eight different cases of two-sided binding with a non-contributing input side. In these cases, three probes grow within a macrotile, and we must take special care to ensure that the probes are coordinated enough to allow for cooperative binding simulation. The eight cases under consideration are given in Figure~\ref{fig:twosided_noncontributing}. In each case we assume that each glue is strength-$1$  and that two of these glues permit cooperative binding while the other glue mismatches or overbinds, whichever the case may be. In general, there are 13 cases of two-sided binding with a non-contributing input side. Five of the eight cases presented here -- (b), (c), (d), (f), and (g) -- are the equivalent up to reflection to the five cases not presented.

\begin{figure}[htp]
\begin{center}
\includegraphics[width=3.0in]{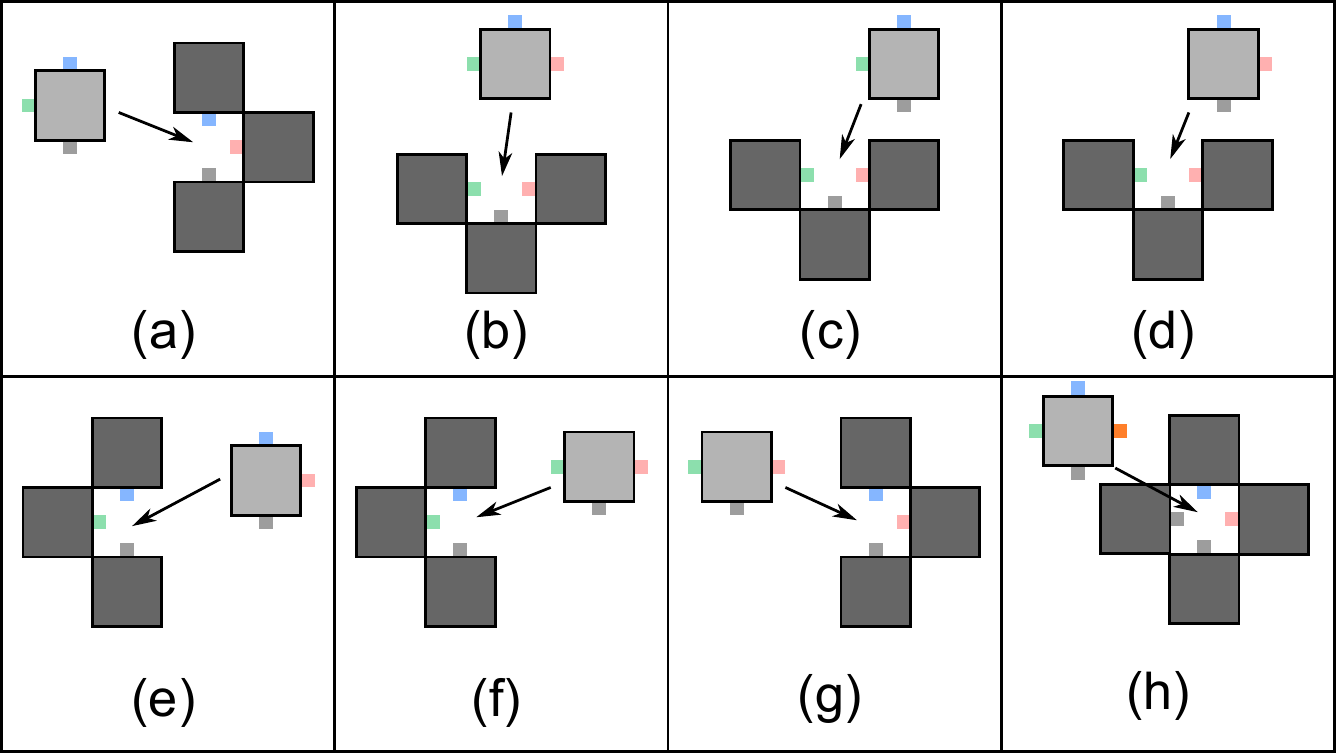}
\caption{Binding via cooperation of two strength-$1$ glues with non-contributing input supersides.}
\label{fig:twosided_noncontributing}
\end{center}
\end{figure}

Figure~\ref{fig:twosided_noncontributing_supertile} gives a schematic image for the simulation of the seven cases of Figure~\ref{fig:twosided_noncontributing}. In each case, two probes meet at a point of cooperation and one of the two cooperator gadgets is used to mimic the cooperation that occurs in aTAM systems. To coordinate these probes, we will also have to use the \emph{crosser} gadget. See Section~\ref{sec:gadgets} for a detailed description of how each of these gadgets works.

\begin{figure}[htp]
\begin{center}
\includegraphics[width=4.5in]{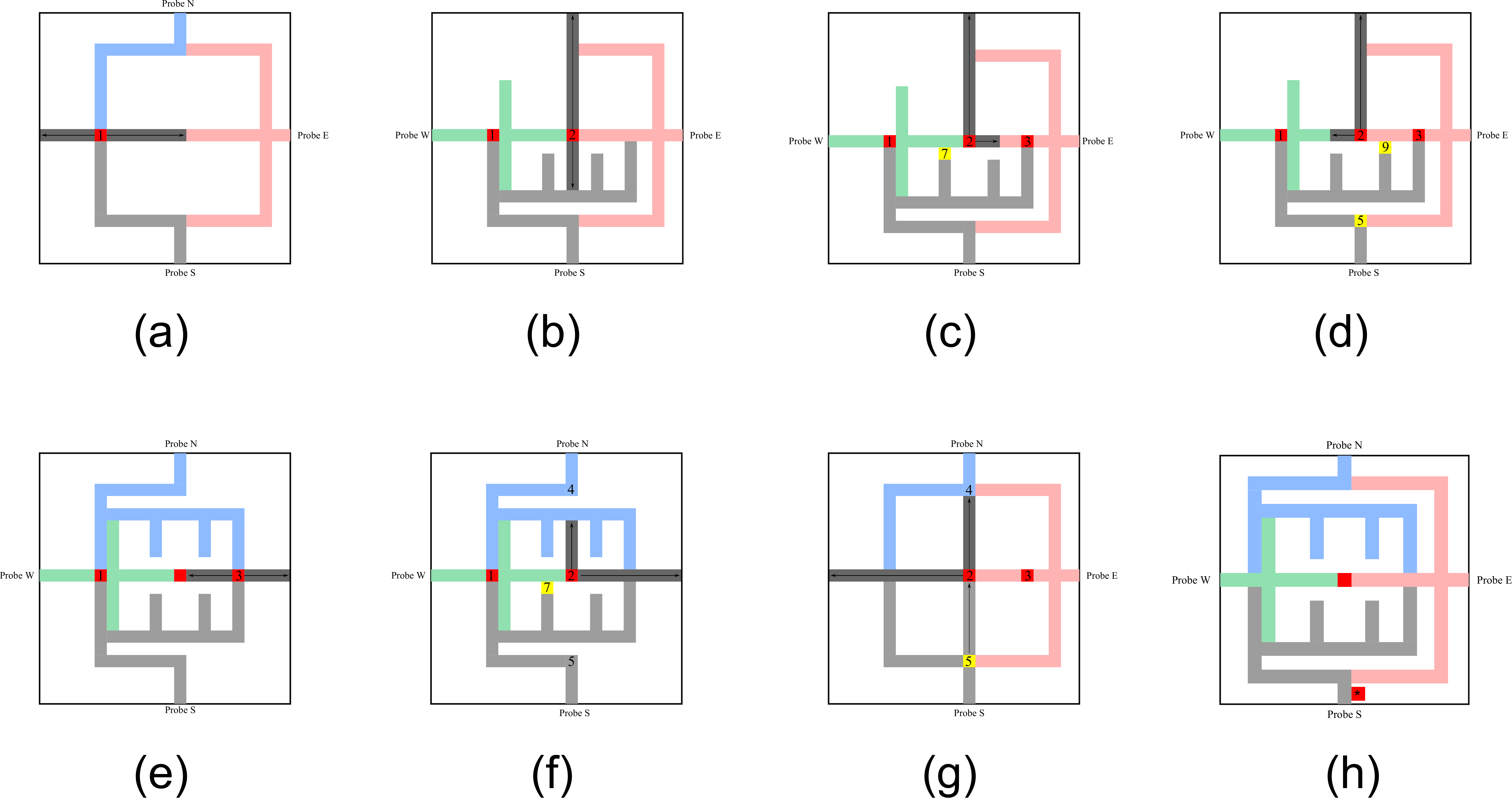}
\caption{Growth of a macrotile that simulates the binding in~\ref{fig:twosided_noncontributing}.}
\label{fig:twosided_noncontributing_supertile}
\end{center}
\end{figure}

\noindent{\bf Case (a):}
Probe N and Probe S meet at POC1 in Figure~\ref{fig:twosided_supertile}(a). A gap cooperator gadget is used to allow for the placement of a duple if and only if there exists a tile type in $T$ with north glue $g_N$ and south glue $g_S$. There is a unique duple for each such tile type in $T$ and the binding of one of these duples allows for the growth of subassemblies that output glues corresponding to the glues on the east and west of the relative tile type in $T$. Notice that since Probe E has started to assemble, the glue outputting subassembly that presents glues on the east side of the macrotile is halted when this subassembly meets the Probe E subassembly.
\newline

\noindent{\bf Case (b):}
Probe E and Probe W simulate cooperative binding when they meet at the POC2 in Figure~\ref{fig:twosided_supertile}(b). A gap cooperator gadget is used to allow for the placement of a duple if and only if there exists a tile type in $T$ with east glue $g_E$ and west glue $g_W$. As in case (a), this duple determines which glue outputting subassemblies form to present glues on the north and south sides of the macrotile. Notice that Probe E and Probe W occupy POC2 and so have automatically won the point of competition at POC2. Therefore, even if Probe S can cooperate with one of the other probes, Probe E and Probe W determine the output glues on the macrotile. Finally, from POC2, subassemblies form to output appropriate glues. The subassembly outputting the south glue of the macrotile will crash into the (at least partially existing) subassembly Probe S.
\newline

\noindent{\bf Case (c):}
Probe S and Probe W simulate cooperative binding as follows. First, Probe W wins POC1 in Figure~\ref{fig:twosided_supertile}(c). After it wins this point, it grows a subassembly to the north and south of POC1 and to the east of where Probe S assembles. This subassembly uses an adjacent cooperator gadget to allow Probe S to assemble secondary probes. One of these probes can cooperate with Probe W at POC7.
An adjacent cooperator gadget is used to allow for the placement of a duple if and only if there exists a tile type in $T$ with west glue $g_W$ and south glue $g_S$. Again, there is a unique duple for each such tile type in $T$ and the binding of one of these duples allows for the growth of subassemblies that output glues corresponding to the glues on the east and north of the relative tile type in $T$. It is at POC7 that a duple is placed that determines which east and north glues to present. Once this duple is placed, growth continues toward POC2. Upon winning POC2, glue outputing subassemblies form that present glues on the east and north sides of the macrotile, and the subassembly presenting the east glue crashes into the subassembly Probe E.
\newline

\noindent{\bf Case (d):}
In this case two assembly sequences can lead to the simulation of this binding type. First, Probe E and Probe S can cooperate at POC5 in Figure~\ref{fig:twosided_noncontributing_supertile}(d), assemble toward POC2  and win POC2. This case is similar to case (d) in Figure~\ref{fig:twosided_supertile}, however, in this case the subassembly presenting glues on the west side of the macrotile is halted by Probe W and only a glue to the north side of the macrotile is presented. Note that it could be the case that secondary probes of Probe S block the assembly that grows from POC5 to POC2. In this case, Probe E and Probe S should still be able to simulate cooperation. They achieve this by cooperating at POC9. At this point, once POC2 is won, glue outputting subassemblies form. This is the case presented in Figure~\ref{fig:twosided_noncontributing_supertile}(d).
\newline

\noindent{\bf Case (e):}
In this case, as in case (d), two assembly sequences can lead to the simulation of this binding type. First, Probe N and Probe S can cooperate at POC1 in Figure~\ref{fig:twosided_noncontributing_supertile}(e), race toward POC2 and win POC2. This case is similar to case (a) in Figure~\ref{fig:twosided_supertile}, however, in this case the subassembly presenting glues on the west side of the macrotile is halted by Probe W and only a glue to the east side of the macrotile is presented. In the case that Probe W wins POC1, secondary probes forms -- one set of secondary probes forms from Probe S and another from Probe N. In this case, Probe N and Probe S should still be able to simulate cooperation. They achieve this by cooperating at POC3. At this point, a glue outputting subassembly forms to present a glue on the east side of the macrotile, since it is known that all other sides have grown input probes. This is the case presented in Figure~\ref{fig:twosided_noncontributing_supertile}(e).
\newline

\noindent{\bf Case (f):}
In this case, first Probe W wins POC1 by assembling a crosser gadget to grow between Probe N and Probe S. In the case where Probe N and Probe S have formed a path of adjacent tiles from the north side of the macrotile to the south side, the crosser gadget detaches a section of this path so that Probe W can assemble. Then Probe W triggers the growth of secondary probes on both Probe N and Probe S. At POC7, an adjacent cooperator gadget assembled from Probe W and Probe S allows for the placement of a duple that determines which glues are output to the north and east sides of the macrotile. Assembly proceeds from POC7 to POC2. Upon winning POC2, an appropriate glue outputting subassembly forms that presents glues to the east side of the macrotile.
\newline

\noindent{\bf Case (g):}
This case is similar to case (d) in Figure~\ref{fig:twosided_supertile} except that the subassembly that presents glues on the north side of the macrotile in case (d) of Figure~\ref{fig:twosided_supertile} crashes into the (at least partially) existing subassembly Probe N.

\noindent{\bf Case (h):}
Up to this point, for simplicity, we have neglected the special point of competition which we look at in the case of two non-contributing sides.  In this case, growth of the macrotile will be similar to the case of one non-contributing side except for the case where $g_N$ and $g_S$ can cooperatively place a tile but no other glues can.  In this particular case, it could be the case that probes representing glues $g_E$ and $g_W$ arrive before either Probe N or Probe S.  Notice that this prevents Probe N and Probe S from cooperatively placing a tile.  In order to handle this peculiar case, we enumerate all of the tiles which $g_N$ and $g_S$ can cooperatively place (this is at most $|T|$) by a function $F$.  In the region labeled $*$ in Figure~\ref{fig:twosided_noncontributing_supertile}(h), we always place a tile from $E \subset S$, which consists of tiles labeled 1 through $T$, nondeterministically.  Let $n$ be the number of tiles $g_N$ and $g_S$ can cooperatively place, and suppose $r$ is the value contained in the label of the tile placed in the $*$ region.  The representation function maps the macrotile to the tile in $T$ given by $F(r \mod n)$.  Since, in this case, the macrotile being assembled is surrounded on all four of its sides, there is not a need for any output subassemblies to be placed.  Furthermore, it should be noted that this region is a ``last resort'' for the representation function.  If there is an appropriate tile placed at any of the other POC regions, the output of the representation function will depend on that tile.

Figure~\ref{fig:macro_detailed} shows an assembled macrotile that simulates the binding which takes place in Figure~\ref{fig:twosided_noncontributing}(d) in the manner shown in Figure~\ref{fig:twosided_noncontributing_supertile}(d).  The blue tiles are part of the subassembly which compose Probe W, the green tiles compose Probe S, and the pink tiles make up Probe E.  All of these probes enter the macrotile region in the direction and location indicated by the arrows. The yellow tiles in the figure show tiles that are cooperatively placed by Probe S and Probe E.  Notice that the growth of the yellow tile placed near the bottom of the figure has been blocked by arm 2 of Probe S, but the second yellow tile in the figure is placed and able to growth a path to POC2.  The dark red tile placed in POC2 starts the growth of an outputting subassembly which grows a new probe into the region north of the macrotile.

\begin{figure}[htp]
\begin{center}
\includegraphics[width=4.0in]{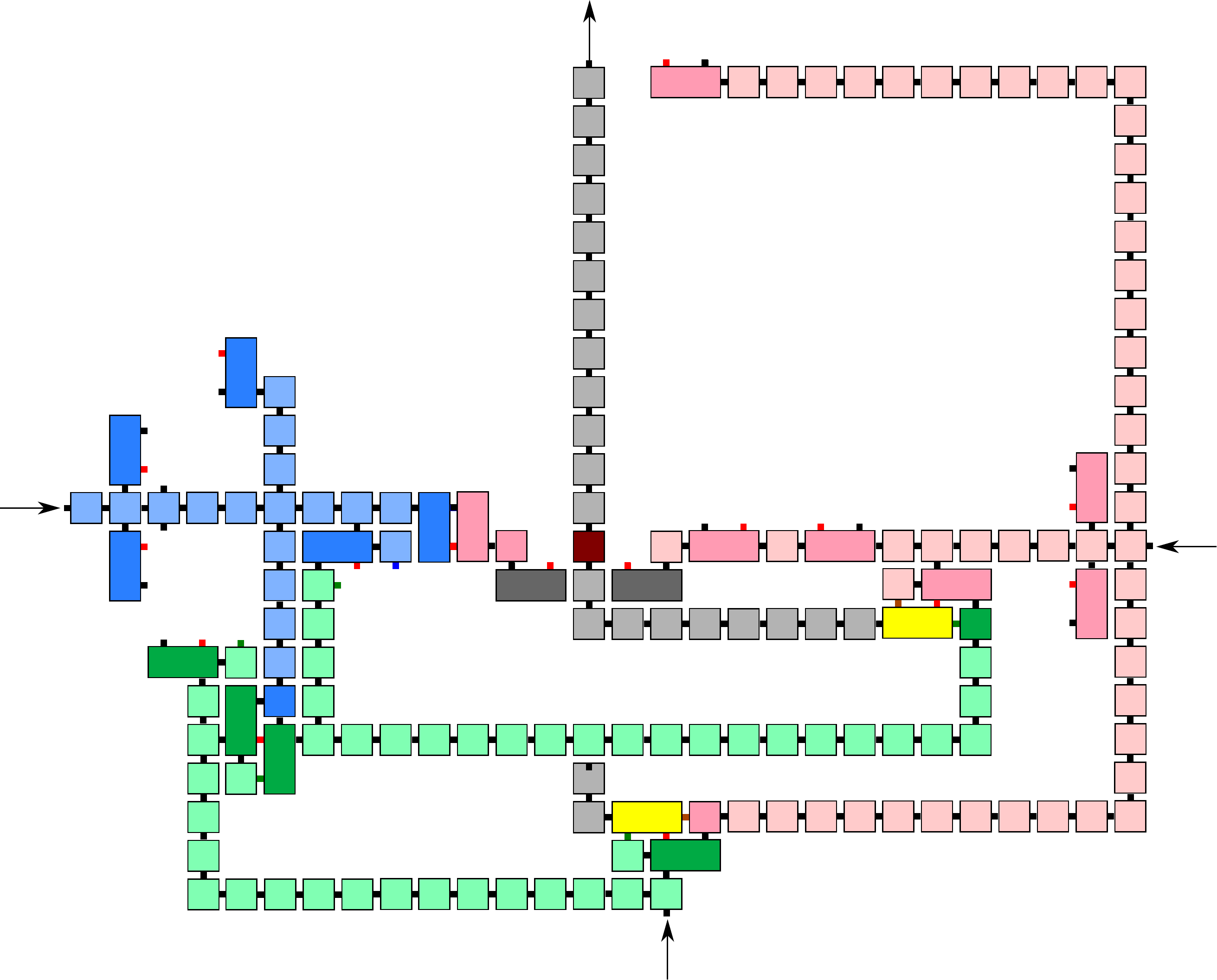}
\caption{A detailed macrotile simulating the binding shown in Figure~\ref{fig:twosided_noncontributing}.}
\label{fig:macro_detailed}
\end{center}
\end{figure}

\section{Seed Formation}\label{sec:seed_form}
In order to complete our description of the simulation of $\calT$, it is necessary to describe the construction of $\sigma'$.  For all $t \in \sigma$, we create special output macrotiles.  An example $\sigma$ is shown in Figure~\ref{fig:seed_form_exT} and the corresponding $\sigma'$ is shown in Figure~\ref{fig:seed_form_exS}.  For these macrotiles, the tile which the representation function depends upon is in the center of the macrotile.  Assembly begins with the macrotiles growing probes for each exposed glue.  If a tile has a side which does not have a glue, then a probe is not grown on that side.  Growth of the assembly then proceeds as described in Section~\ref{sec:poc_repr}.

\begin{figure}[htp]
\begin{center}
\includegraphics[height=1.0in]{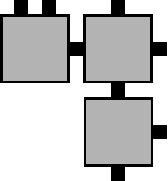}
\caption{An example seed in $\cal T$.}
\label{fig:seed_form_exT}
\end{center}
\end{figure}

\begin{figure}[htp]
\begin{center}
\includegraphics[width=1.5in]{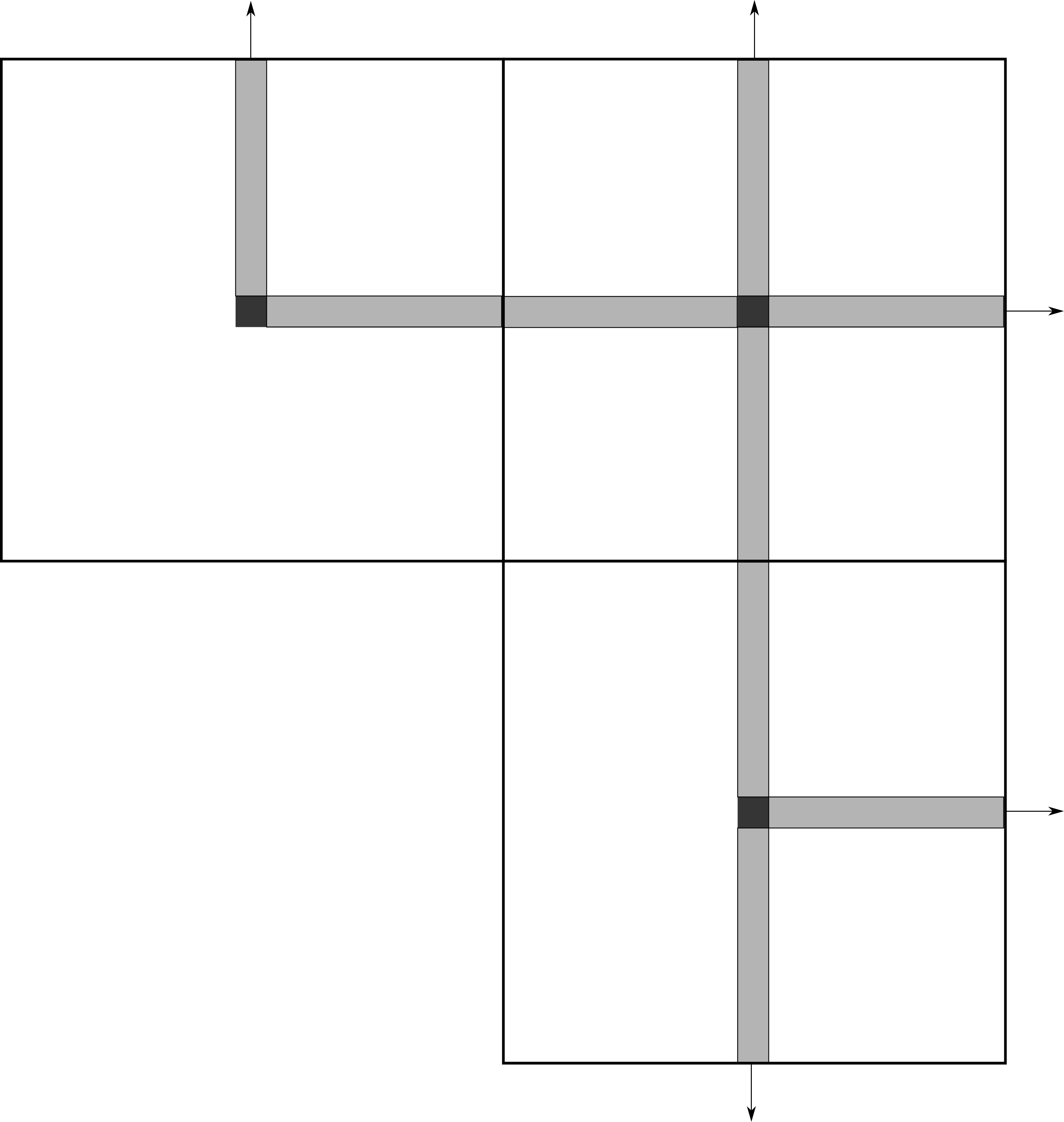}
\caption{The corresponding seed in $\mathcal{S}$ for the seed shown in Figure~\ref{fig:seed_form_exT}.  The arrows indicate in which directions the probes will grow when assembly begins. The black squares represent macrotile regions.}
\label{fig:seed_form_exS}
\end{center}
\end{figure}

\section{Proof of Correctness}\label{sec:proof_of_correctness}
\begin{proof}
Let $\calT = (T, \sigma, 2)$ be an aTAM system and let $\mathcal{S} = (T_S, S, D, \sigma_S)$ be the DrgTAS system obtained from $\calT$ by the construction given in Section~\ref{sec:DrgTAM_sim}.
To show that $\mathcal{S}$ gives a valid simulation of $\calT$ we will use the representation function $R$ described in Section~\ref{sec:poc_repr} and denote the scale factor of this simulation by $m$. We will show the following. 1. Under $R$, $\mathcal{S}$ and $\calT$ have equivalent production. This essentially follows from the construction. 2. Under $R$, $\mathcal{S}$ and $\calT$ have equivalent dynamics. To show this, we must show that when detachment occurs in $\mathcal{S}$ due to a cut of strength less than $1$, the two assemblies produced are of the form that one of them still maps correctly to a represented assembly in $\calT$ while the other produced assembly maps to the empty tile. We must also show that any assembly sequence starting from the latter assembly yields an assembly that maps to the empty tile under $R$.

To show that $\mathcal{S}$ and $\calT$ have equivalent production under $R$, we will first show that $R^*$ maps cleanly.  To see this, note that the probes described in the construction (Section~\ref{sec:poc_repr}) can only be grown from adjacent macrotiles on sides where they placed an outputting subassembly. The probes grown in macrotile regions are never interpreted under $R$ as a tile in $\calT$ until a POC region is won (using cooperation if necessary).  It follows from the construction that macrotile regions which map to the empty tile under $R$ will not grow any outputting subassemblies which can initiate probes in adjacent macrotile locations until a POC region is won and the macrotile first maps to a tile in $\calT$. Therefore, $R^*$ maps cleanly.

Now, to see that $\left\{R^*(\alpha') | \alpha' \in \prodasm{S}\right\} = \prodasm{T}$, let $\alpha'$ be in $\prodasm{S}$. Then by the construction, any $m$-block macrotile, $B$, in $\alpha'$ maps to the empty tile or to some tile type in $T$, and only maps to a tile in $T$ if $B$ is part of the seed of $\mathcal{S}$, or adjacent $m$-block macrotiles of $B$ in $\alpha'$ expose glues that allow for the growth of $B$. In the latter case, the construction shows that $B$ can only map to a tile type whose glues match the glues represented on the sides of adjacent $m$-block macrotiles. Since this holds for any $m$-block macrotile in $\alpha'$, $R^*(\alpha')$ is in $\prodasm{T}$. This shows that $\left\{R^*(\alpha') | \alpha' \in \prodasm{S}\right\} \subseteq \prodasm{T}$. Then, for $\alpha\in \prodasm{T}$ and an assembly sequence $\vec{\alpha}$ resulting in $\alpha$, the construction also shows that we can grow $m$-block macrotiles following the assembly sequence $\vec{\alpha}$ to obtain an $\alpha'$ in $\prodasm{S}$ such that $R^*(\alpha') = \alpha$. Therefore, we also have $\left\{R^*(\alpha') | \alpha' \in \prodasm{S}\right\} \supseteq \prodasm{T}$. Thus, $\left\{R^*(\alpha') | \alpha' \in \prodasm{S}\right\} = \prodasm{T}$.

To show that $\mathcal{S}$ and $\calT$ have equivalent dynamics, first note that the construction implies
that $\alpha' \rightarrow_{+}^\mathcal{S} \beta'$ if and only if $R^*(\alpha') \rightarrow R^*(\beta')$. To see this note that when a single tile (or duple) is added to $\alpha'$, if the tile (or duple) does not win a point of competition, then $R^*(\alpha') = R^*(\beta')$. On the other hand, if the tile (or duple) does win a point of competition, the macrotile is determined once and for all. Moreover, an assembly in a macrotile region cannot map to a tile type $t$ under $R$ unless some adjacent macrotile region (or regions in the case of simulation of cooperation) maps to a tile type with a glue (or glues) that allows for the placement of a tile with type $t$. Therefore,  $R^*(\alpha') \rightarrow R^*(\beta')$.

What is left to show is that for $\alpha'$ in $\prodasm{S}$ such that $R^*(\alpha') = \alpha$, when a cut with strength less than $1$ exists in $\alpha'$, the two assemblies that on each side of the cut are such that one of the assemblies, $\beta'_1$ say, still represents $\alpha$, while the other assembly, $\beta'_2$, represents the empty tile. Moreover, we must show that the result of any assembly sequence starting from $\beta'_2$ must also represent the empty tile. To see this, note that in the only cases where there exists a cut of strength less than $1$ in any of the gadgets given in Section~\ref{sec:gadgets}, the cut separates the assembly into two configurations where one of the configurations is given as one of the assemblies in Figure~\ref{fig:gad_junk}. One can check that the configurations given in Figure~\ref{fig:gad_junk} quickly become terminal and represent the empty tile. The other configuration is an assembly that still represents $\alpha$ since there is never a cut of strength less than $1$ separating points of cooperation or points of competition from an assembly.

To see that the scale factor is $O(1)$, note that the lengths of the paths and sizes of gadgets which make up macrotiles are all fixed, independent of $\calT$.  In order to see that the tile complexity of $|S \cup D|$ is $O(|T|)$ notice that for each tile $t \in T$ there are a bounded number of ways to bind, which means that the different types of tiles that simulate the binding of $t$ is bounded.  Furthermore, observe that for each of these potential ways of binding, the number of tiles required to assemble the macrotile which maps to $t$ under the representation function, given the constant scale factor of macrotiles, is constant.  Consequently, $|S \cup D| = O(|T|)$.
\end{proof}

} %later
\vspace{-5pt}
\begin{corollary}\label{cor:IU}
There exists a DrgTAM tile set $U$ which, at temperature-$1$, is intrinsically universal for the aTAM.  Furthermore, the sets of singletons and duples, $S$ and $D$, created from $U$ are constant across all simulations.
\end{corollary}
\vspace{-5pt}
As mentioned above this result follows from \cite{IUSA}.  See Section~\ref{sec:IU_proof} for more details.

\ifabstract
\later{
\section{Proof of Corollary~\ref{cor:IU}} \label{sec:IU_proof}
\begin{proof}
To prove Corollary~\ref{cor:IU}, we let $\calT = (T,\sigma,\tau)$ be an arbitrary TAS in the aTAM.  Let $U_{\calT} = (U,\sigma_{\calT},2)$ be an aTAM TAS which simulates $\calT$ using the tile set $U$, given in \cite{IUSA}, which is intrinsically universal for the aTAM.  Now let $\mathcal{D} = (T_U, S, D, \sigma_{\calT}', 1)$ be a DrgTAS, constructed as given by the proof of Theorem~\ref{thm:DrgTAS-sim}, which simulates $U_{\calT}$.  We now note that, regardless of $\calT$, the same tile set $U$ is used to simulate it in the aTAM, and that $T_U$, $S$, and $D$ depend only upon the tile set being simulated by $\mathcal{D}$ (i.e. the only thing that changes in $\mathcal{D}$ as $\calT$ changes is $\sigma_{\calT}$).  Therefore, a single tile set $T_U$ suffices to simulate any arbitrary aTAM TAS, and thus $T_U$ is intrinsically universal for the aTAM.
\end{proof}
} %later